\documentclass[10pt]{article}
\usepackage{amsmath, amsfonts, amsthm, amssymb,verbatim}
\newtheorem{theorem}{Theorem}[section] 
\usepackage[a4paper,dvips,nohead,twoside]{geometry}
\newtheorem{definition}[theorem]{Definition}
\newtheorem{lemma}[theorem]{Lemma}

\newtheorem{proposition}[theorem]{Proposition}
\newtheorem{example}[theorem]{Example}

\numberwithin{equation}{section}
\usepackage[ps2pdf, bookmarks=true,
            bookmarksopen=true,colorlinks=false]{hyperref}
\newcommand \widehatu{\widehat u}

\newcommand \bz {\begin{itemize}}
\newcommand \ez {\end{itemize}}
\newcommand \ben {\begin{enumerate}}
\newcommand\een {\end{enumerate}} 
\newcommand \N {\mathbb N}
\newcommand \R {\mathbb R} 
\newcommand{\Eqref}[1]{\eqref{#1}}

\newcommand{\Sectionref}[1]{Section~\ref{#1}}  
\newcommand{\Defref}[1]{Definition~\ref{#1}}
\newcommand{\Lemref}[1]{Lemma~\ref{#1}}
\newcommand{\Propref}[1]{Proposition~\ref{#1}}
\newcommand{\Theoremref}[1]{Theorem~\ref{#1}}

\newcommand \auth {} 
\newcommand \jou {\textit}

\newcommand \muh M 
\newcommand \vv u

\newcommand \RR 		{\mathbb{R}}  
\newcommand \del 	\partial
\newcommand \eps 	\epsilon
\newcommand \lam 	\lambda 
\newcommand \be 		{\begin{equation}}
\newcommand\ee 		{\end{equation}}

\let\oldmarginpar\marginpar
\renewcommand\marginpar[1]{\-\oldmarginpar[\raggedleft\footnotesize #1]%
{\raggedright\footnotesize #1}}

\begin{document}

\title{Second-order hyperbolic Fuchsian systems.\\
  General theory} 

\author{Florian Beyer$^1$ and Philippe G. LeFloch$^2$ 
}

\date{November 10, 2010}

\maketitle 

\footnotetext[1]{
Department of Mathematics and Statistics, University of Otago, P.O.~Box 56, Dunedin 9054, New Zealand, E-mail: {fbeyer@maths.otago.ac.nz}.
\newline
$^2$Laboratoire Jacques-Louis Lions \& Centre National de la Recherche Scientifique, 
Universit\'e Pierre et Marie Curie (Paris 6), 4 Place Jussieu, 75252 Paris, France. E-mail: {pgLeFloch@gmail.com.}
\newline 
\textit{\ AMS Subject Classification.} Primary: 35L10, 83C05. Secondary: 83C75. 
\, \textit{Key words and phrases.} Einstein field equations, Gowdy spacetime, singular partial differential equation, 
Fuchsian equation, canonical expansion, singular initial value problem. 
\newline 
\bf A shortened version of the material presented in this 
preprint is found in:
{\sl F. Beyer and P.G. LeFloch, Second-order hyperbolic Fuchsian systems and applications,
Class. Quantum Grav. 27  (2010), 245012.}   
}

\begin{abstract}  
We introduce a class of singular partial differential equations, the 
{\sl second-order hyperbolic Fuchsian systems,}  
and we investigate the associated initial value problem when data are imposed on the singularity. 
First of all, we analyze a class of equations in which hyperbolicity is not assumed 
and we construct asymptotic solutions of arbitrary order. 
Second, for the proposed class of second-order hyperbolic Fuchsian systems, 
we establish the existence of solutions with prescribed asymptotic behavior on the singularity. 
Our proof is based on a new scheme which is also suitable to design numerical approximations.  
Furthermore, as shown in a follow-up paper, the second-order Fuchsian framework is
appropriate to handle Einstein's field equations for Gowdy symmetric
spacetimes and allows us to recover (and slightly generalize) earlier
results by Rendall and collaborators, while providing a direct
approach leading to accurate numerical solutions. The proposed framework is also robust enough to encompass matter models
arising in general relativity.
\end{abstract}
 
\tableofcontents

\section{Introduction}

This is the first of a series of papers \cite{BeyerLeFloch2,BeyerLeFloch3}
devoted to solving the initial value problem for certain classes 
of spacetimes of general relativity. Specifically, we are interested in spacetimes enjoying certain symmetries, 
especially the Gowdy symmetry, and in a formulation when data are imposed on a singular hypersurface where curvature generically blows-up. 
  
For instance, one may consider $(3+1)$-dimensional, vacuum spacetimes $(M,g)$ with spatial topology $T^3$,  
satisfying the vacuum Einstein equations under the Gowdy symmetry assumption, i.e.~under 
the existence of an Abelian $T^2$ isometry group with spacelike orbits whose so-called ``twist constants'' \cite{Chrusciel}
vanish identically. These so-called Gowdy spacetimes on $T^3$ were first studied in \cite{Gowdy73}. 
A combination of theoretical and numerical works has led to a detailed picture of the behavior of solutions 
to the Einstein equations as one approaches the singular boundary of such spacetimes; 
see \cite{Chrusciel, IsenbergMoncrief, Moncrief, Ringstrom2, Ringstrom4, Ringstrom6}.

For the above analysis, one important tool was provided by 
Rendall and his collaborators \cite{AnderssonRendall, KichenassamyRendall, Rendall00} who 
developed the so-called {\sl Fuchsian method} to handle the singular evolution equations associated 
with the Einstein equations for such spacetimes. 
This method allowed the authors to derive precise information about the behavior of solutions near the singularity, 
which was a key step in the general proof of Penrose's strong cosmic conjecture eventually 
established by Ringstr\"om \cite{Ringstrom6}. 

Our aim here and in the follow-up papers \cite{BeyerLeFloch2,BeyerLeFloch3} is two-fold. 
On one hand, we re-visit Rendall's theory (which covers smooth solutions to first-order equations) 
and we develop here a well-posedness theory in Sobolev spaces for the class of 
{\sl second-order hyperbolic Fuchsian systems,} defined below in Section~3. 
As we will show, this class includes many systems of equations arising in general relativity. 

On the other hand, following a strategy initially proposed in Amorim, Bernardi, and LeFloch \cite{ABL}, 
we are interested in the numerical approximation of these Fuchsian equations  
when data are imposed on the singularity of the spacetime and one evolves the solution {\sl from}
 the singularity.
In contrast, standard numerical approaches consider the evolution {\sl toward} the singularity.  
Our main improvement here upon \cite{ABL} is that, in short, no restriction need be imposed on the coefficients of the Fuchsian system, 
provided asymptotic expansions of arbitrary large order are sought. This issue will be developed in \cite{BeyerLeFloch2}. 

The present paper is theoretical in nature, and our main contributions can be summarized as follows: 

\begin{itemize}

\item {\bf Second-order formulation.} 
 First-order Fuchsian systems have been used successfully to handle the equations describing (vacuum) 
Gowdy spacetimes \cite{AnderssonRendall, KichenassamyRendall, Rendall00}. However, we argue here that 
second-order hyperbolic Fuchsian systems, as we define them in this paper, arise more naturally in the applications. 
For instance, Gowdy spacetimes described by second-order partial differential equations, and 
it is natural to keep the second-order structure. In particular, expansions required in the theoretical analysis  
(as well as in actual computations required for the numerical 
discretization) are also more natural with the second-order formulation. 

\item {\bf Hyperbolicity property.}  In addition, in the applications to general relativity, 
ensuring and checking the hyperbolicity of the equations under consideration (after a suitable 
reduction of the original equations) is expected
  to be more convenient with the second-order formulation. 
  For instance, it is easier to recognize from the original second-order system if the equations form
  a system of coupled wave equations, while this property is
  much less evident in the first-order formulation. As the
  discussion in \cite{Rendall00} demonstrates, it can be
  cumbersome to formulate a system of coupled non-linear wave
  equations as a first-order hyperbolic Fuchsian system satisfying all the properties required for local well-posedness.

\item {\bf Singular part.} We also introduce here a construction algorithm which includes 
the singular part of the solution. This is different from the classical first-order
  Fuchsian approach, where one first makes an ansatz by removing from the solution 
  its (expected) singular part, and then
  applies the Fuchsian theory to the (regular) remainder. Our approach leads to the notion of the singular 
  initial value problem for Fuchsian hyperbolic equations which can be understood as a generalization 
  of the initial value problem for (standard) hyperbolic equations.  
  This singular initial value problem covers cases whose singularity
  is oscillatory in nature, in a manner consistent with the so-called BKL conjecture (introduced by 
  Belinsky, Khalatnikov, and Lifshitz). 
  
\end{itemize}

We shall, first, define the class of systems of interest, that is, the second-order (hyperbolic) Fuchsian systems
and, then, investigate the associated initial value problem when data are imposed on the singularity. 
Precisely, in Section~2, we analyze a class of equations in which hyperbolicity is not assumed 
and we construct asymptotic solutions of arbitrary order. In Section~3, we treat
the proposed class of second-order hyperbolic Fuchsian systems, and 
establish the existence of solutions with prescribed asymptotic behavior on the singularity. 

In \cite{BeyerLeFloch2}, we will apply our Fuchsian framework and treat the class of Gowdy spacetimes, and this 
will allow us to recover (and slightly generalize) earlier results by Rendall and collaborators, 
while providing a more direct approach. Although conceptually similar, the proposed analysis based on 
second-order equations lead to a simpler description and provides a definite advantage for the applications. 
Moreover, as we will demonstrate in \cite{BeyerLeFloch2}.  
the approach introduced in the present paper can be cast into 
a discretization scheme and allows us to numerically and accurately compute solutions to the
initial value problem. Our theory also turns out to be robust enough to extend to matter models, 
especially to the Einstein-Euler equations \cite{BLSS, LeFlochRendall,LeFlochStewart,LeFlochStewart2}, 
as discussed in \cite{BeyerLeFloch3}.


\section{Second-order Fuchsian systems}
\label{sec:firstorderFuchsian}

\subsection{Terminology and objectives}

In this section, we rely mainly on techniques for ordinary
differential equations (ODE's).  The main theory of interest developed
in the next section (Section~3) will require an hyperbolicity assumption
which is not yet made at this stage.  The main purpose of this section
is to present some important terminology and concepts within the
simple framework of ODE's, and to point out some difficulty arising with
singular equations.

First of all, $t \geq 0$ denoting the time variable, the operator 
\[
D:=t\del_t
\] 
will often be used, rather than the partial derivative
$\del_t$. Indeed, the weight $t$ is convenient to handle asymptotic
expansions near the singularity $t=0$.  Occasionally, we write $D_t$
instead of $D$, especially when several time variables are involved.

\begin{definition}[Second-order Fuchsian systems]
\label{def:Fuchsian}
A second-order Fuchsian system is a system of partial
differential equations of the form
\be
    \label{eq:secondorderFuchsian}
    D^2 u(t,x)+2A(x) \, D u(t,x)+B(x) \,  u(t,x) 
    = f[u](t,x) 
\ee
with unknown function $u:(0,\delta]\times U\rightarrow \R^n$ (for some
$\delta>0$ and interval $U\subset\R$), where the coefficients $A=A(x)$
and $B=B(x)$ are diagonal,
 $n\times n$ matrix-valued maps defined on $U$,
 and the source-term $f=f[u](t,x)$ is an 
$n$-vector-valued map of the form 
\be
  \label{eq:inhomog}
  f[u](t,x):=f(t, x, u, Du, \del_x u, \del_x Du,\ldots,
  \del_x^k u, \del_x^k Du),
\ee
for some integer $k \geq 0$.
\end{definition}
 
We assume that the coefficients $A$ and $B$ do not depend on $t$ and
we derive our results under this assumption. The generalization to
time-dependent coefficients does not bring essential difficulties.
The assumption that the matrices $A$, $B$ are diagonal is not a
genuine restriction, if the system with arbitrary matrices $A$, $B$
can be recast in diagonal form, i.e.\ if $A$, $B$ admit a common basis
of eigenvectors. Under this condition 
(which is always satisfied in the applications of interest in general relativity), the system is 
``essentially decoupled'' since the coupling takes place in
terms of non-leading order, only.

We denote the eigenvalues
of $A$ and $B$ by  $a^{(1)},\ldots,a^{(n)}$ and
$b^{(1)},\ldots,b^{(n)}$, respectively. When it is not necessary to
specify the superscripts, we just write $a, b$ to denote any eigenvalues of $A,B$. 
With this convention, we introduce: 
\be
  \label{eq:deflambda2}
    \lambda_{1}:=a+\sqrt{a^2-b},\quad \lambda_{2}:=a-\sqrt{a^2-b}.
\ee
It will turn out that these coefficients describe the expected
behavior at $t=0$ of general solutions to \Eqref{eq:secondorderFuchsian}. 

In \Defref{def:Fuchsian}, the assumption that $U$ is a one-dimensional
domain makes the presentation simpler, but most results below remain
valid for arbitrary spatial dimensions.  For definiteness and without
much loss of generality, we assume throughout this paper that all
functions under consideration are periodic in the spatial variable $x$
and that $U$ is the periodicity domain. All data and solutions are
extended by periodicity outside the interval $U$.

The left-hand side of \eqref{eq:secondorderFuchsian} is referred to as
the principal part of the system. The reason for incorporating certain
lower derivative terms in the principal part is that we expect these
terms to be of the same leading-order at the singularity $t=0$. In
contrast, the source-term is anticipated as negligible in some sense
there, see below. Observe that, at this level of generality, there is
some freedom in bringing terms from the principal part to the
right-hand side, and absorbing them into the source-function $f$ (or
vice-versa). This freedom has several (interesting) consequences,
as we will discuss later on: roughly speaking, some normalization will
be necessary later, yet at this stage, we do not fix the behavior of
$f$ at $t=0$.

We are mainly interested in solving a {\sl singular} initial value
problem associated with \Eqref{eq:secondorderFuchsian}, with data
prescribed on the singularity $t=0$, in a sense made precise later on.
The fundamental question is, of course, to determine conditions on the
data and coefficients ensuring existence and uniqueness of a solution
$u$.  It will turn out that the behavior at $t=0$ cannot be prescribed
arbitrarily, but is tight to the value of the coefficients
$\lambda_{1}$ and $\lambda_{2}$ defined in \eqref{eq:deflambda2}.
Indeed, we will specify the behavior of solutions at $t=0$, in terms
of freely specifiable functions (the data on the singularity), and
derive an asymptotic expression of arbitrary order providing the
asymptotic form of general solutions.


\subsection{The case of linear ODE's depending on
  \texorpdfstring{$x$}{x} as a parameter}
\label{sec:linearODEcase}

\subsubsection*{Explicit formula}

We begin our investigation of the second-order Fuchsian systems
\Eqref{eq:secondorderFuchsian} by treating the case $f=w(t,x)$ for
some given function $w$.  We are led to consider a family of scalar
ordinary differential equations that are completely independent from
each other; recall that the matrices $A$ and $B$ of the principal part
are diagonal.  Without loss of generality, we thus assume that
$n=1$ throughout the present section.  In turn, the spatial variable
$x$ is treated as a {\sl parameter} which we do not need to
write explicitly yet.  As we will see, it is instructive to express
the general solution in this elementary case, as we now do.

Consider the following singular, inhomogeneous, singular ordinary differential equation
\be
    \label{eq:inhomEq}
    D^2 u(t)+2a\,D u(t)+b\, u(t)= w(t) 
\ee
with unknown $u=u(t)$, where $w=w(t)$ is a given locally integrable function. Further
integrability of $w$ near $t=0$ will be imposed shortly below. 
Recall that $a$ and $b$ are constant in $t$ and, 
$\lambda_1, \lambda_2$ were defined in \Eqref{eq:deflambda2}.  We
begin with a formal result and the convergence of the integral terms
will be discussed rigorously later.

  \begin{proposition}[Linear second-order Fuchsian ODE. Formal version]
    \label{prop:explicitexpr} 
   General solutions of the inhomogeneous 
  singular ordinary differential equation \Eqref{eq:inhomEq} are given by 
\be
    \label{eq:inhomSol}
    u(t)=
    \begin{cases}
      u_*\,t^{-a}\,\ln t+u_{**}\,t^{-a}
      + \displaystyle \int_1^{\infty}w(t/\zeta)\zeta^{-a-1}\ln\zeta\, d\zeta, 
      & \quad a^2=b,\\
      u_*\,t^{-\lambda_1}+u_{**}\,t^{-\lambda_2} & \\
      \quad+\frac{1}{\lambda_1-\lambda_2}
      \displaystyle  \int_1^{\infty}w(t/\zeta)
      \left(\zeta^{-\lambda_2-1}-\zeta^{-\lambda_1-1}\right)      
      d\zeta,
      & \quad a^2 \neq b, 
    \end{cases}
\ee
in which $u_*$ and $u_{**}$ are prescribed data. 
Alternatively, one can write \eqref{eq:inhomSol} in the form 
\[  
    u(t)=
    \begin{cases}
      u_*\,t^{-a}\,\ln t+u_{**}\,t^{-a}
      + t^{-a}\displaystyle \int_0^{t}w(s)s^{a-1}\ln\frac ts\, ds, 
      & \quad a^2=b,\\
      u_*\,t^{-\lambda_1}+u_{**}\,t^{-\lambda_2} & \\
      \quad+\frac{1}{\lambda_1-\lambda_2}
      \displaystyle\left( 
        t^{-\lambda_2}\int_0^{t}w(s)s^{\lambda_2-1}ds
        -t^{-\lambda_1}\int_0^{t}w(s)s^{\lambda_1-1}\right),
      & \quad a^2 \neq b.
    \end{cases}
\]
\end{proposition} 

\begin{proof} In the rescaled time variable $\eta:=-\ln t$, 
equation \Eqref{eq:inhomEq} has constant coefficients, indeed 
\be
\label{eq:inhomEqeta}
\widehatu''(\eta)-2a\,\widehatu'(\eta) + b\, \widehatu(\eta) = \widehat w(\eta), \ee
where $\widehatu(\eta):=u(e^{-\eta})$ and $\widehat w(\eta):=w(e^{-\eta})$
and the prime $'$ denotes a derivative with respect to $\eta$.  This is nothing but 
a linear harmonic oscillator equation with friction term $-2a\,
\widehatu'$ and forcing term $\widehat w$. The singularity of
\Eqref{eq:inhomEq} at $t=0$ corresponds to the singularity at infinity $\eta=\infty$.

First, we seek for general solutions of the homogeneous equation for
$\widehat w\equiv 0$. From the ansatz $\widehatu = e^{\lambda\eta}$ we
obviously get the roots defined in \eqref{eq:deflambda2}.  The general
solution of the homogeneous equation is thus
\[
\widehatu(\eta)=
\begin{cases}
  -u_*\,\eta\,e^{a\eta}+u_{**}\,e^{a\eta}, & \quad a^2=b,\\
  u_*\,e^{\lambda_1\eta}+u_{**}\,e^{\lambda_2\eta}, & \quad a^2 \neq b,
\end{cases}
\]
where $u_*$ and $u_{**}$ are constants with respect to $\eta$ (and the
negative sign is chosen for convenience in the following discussion).

Particular solutions of the general inhomogeneous equation
\Eqref{eq:inhomEqeta} are easily constructed by the Duhamel
principle. Let $\widetilde u := \widetilde u(\eta)$ be the solution of
the homogeneous equation for vanishing data $\widetilde u(0)=0$ and
${\widetilde u}'(0)=-1$. For any given $\widehat w$, the function
\[\widehatu(\eta)=\int_{-\infty}^0 \widetilde u(\tau) \, \widehat w(\eta-\tau) \, d\tau
\] 
is a particular solution of the inhomogeneous equation.  (This is true
only formally at this stage, since we have not yet checked under which
conditions the integral exists and can be differentiated.) Hence, the
general solution $u$ of \Eqref{eq:inhomEqeta} reads
\[\widehatu(\eta)=
\begin{cases}
  - u_*\,\eta\,e^{a\eta}+u_{**}\,e^{a\eta}
  -\int_{-\infty}^0 \tau e^{a\tau}\widehat w(\eta-\tau)d\tau, 
  & \quad a^2=b,\\
  u_*\,e^{\lambda_1\eta}+u_{**}\,e^{\lambda_2\eta} & \\
  \quad-\frac{1}{\lambda_1-\lambda_2}
  \int_{-\infty}^0 \left(e^{\lambda_1\tau}-e^{\lambda_2\tau}\right)
  \widehat w(\eta-\tau)d\tau, 
  &  \quad a^2 \neq b.
\end{cases}
\] 
  
Returning to the original time $t>0$, we find
\begin{align*}
  u(t)&=
  \begin{cases}
    u_*\,t^{-a}\,\ln t+u_{**}\,t^{-a}
    -\int_{\infty}^1 (-\ln\zeta) \zeta^{-a}\widehat w(\ln(\zeta/t))
    (-1/\zeta)d\zeta, 
    & \quad a^2=b,\\
    u_*\,t^{-\lambda_1}+u_{**}\,t^{-\lambda_2} & \\
    \quad\quad-\frac{1}{\lambda_1-\lambda_2}
    \int_{\infty}^1\left(\zeta^{-\lambda_1}-\zeta^{-\lambda_2}\right)
    \widehat w(\ln(\zeta/t))
    (-1/\zeta)d\zeta,
    & \quad a^2 \neq b, 
  \end{cases}
\end{align*}
and by substituting $\tau=-\ln\zeta$, this concludes the proof. 
\end{proof}

\subsubsection*{The spatial coordinate \texorpdfstring{$x$}{x} as a parameter}
For the later discussion, it is convenient to write the spatial
variable $x$ explicitly as a parameter now.  Define
$\Gamma(x):=\sqrt{a(x)^2-b(x)}$ which might be real or imaginary
dependent on the values of the coefficients. If there are points
$x_0\in U$ so that $\Gamma(x_0)=0$ and other points $x_1\in U$ with
$\Gamma(x_1)\not=0$, then we will renormalize the coefficients
$u_*(x)$ and $u_{**}(x)$ in \Eqref{eq:inhomSol} as follows. In order
to obtain a continuous transition from the non-degenerate case
$\Gamma\not=0$ to the degenerate case $\Gamma=0$, let us first rename
the coefficient functions for the case $a^2\not=b$ in
\Eqref{eq:inhomSol} to $\hat u_*$ and $\hat u_{**}$. Now if we set
\begin{equation}
  \label{eq:transitiondata}
  \hat u_*(x)=\frac{u_*(x)-u_{**}(x)/\Gamma(x)}2,\quad
  \hat u_{**}(x)=\frac{u_*(x)+u_{**}(x)/\Gamma(x)}2,
\end{equation} 
and choose $u_*(x)$, $u_{**}(x)$ to be, say, continuous functions,
then the function determined by the two leading terms in
\Eqref{eq:intuitiveexpansion} is continuous in $x$ for all $t>0$ even
at $x=x_0$, provided $\Gamma$ is continuous. Indeed the full general
solution $u(t,x)$ of \Eqref{eq:inhomEq} is continuous in $x$ for all
$t>0$ in this case.

In view of
Proposition~\ref{prop:explicitexpr} it is natural to define the
solution operator $H$ associated with a source function $w=w(t,x)$ by
\begin{equation}
  \label{eq:defH}
  (H[w])(t,x):=\begin{cases}
    \displaystyle t^{-a(x)} \int_0^{t}w(s,x)s^{a(x)-1}\ln\frac ts\, ds, 
    & \quad (a(x))^2 = b(x),\\
    \displaystyle\frac{1}{\lambda_1(x)-\lambda_2(x)}
    \Bigl( 
    t^{-\lambda_2(x)}\int_0^{t}w(s)s^{\lambda_2(x)-1}ds & \\ 
    \qquad\qquad\qquad
    \displaystyle-t^{-\lambda_1(x)}\int_0^{t}w(s,x)s^{\lambda_1(x)-1}\Bigr), &
    \quad (a(x))^2 \neq b(x). 
  \end{cases}
\end{equation}
It represents the solution of \Eqref{eq:inhomEq} for the choice
$u_*=u_{**}=0$, at least on the formal level so far. According to the
previous discussion, the non-degenerate case $a^2\not=b$ in the
definition converges to the degenerate case $a^2=b$ at degenerate
points continuously, provided the coefficients are continuous, and
vice versa.

Fixing some $\delta>0$, we now assume that the source $w$ in
\eqref{eq:inhomEq} belongs to $C^{l\times m}((0,\delta]\times U)$,
that is, $w$ is $l$-times continuously differentiable with respect $t$
and $m$-times continuously differentiable with respect to $x$ on
$(0,\delta]\times U \subset \RR^2$. Here, $l$ and $m$ are non-negative
integers. Moreover, we assume that the coefficients $a$ and $b$ of the
equation are $C^m(U)$. In this case, the general theory of ordinary
differential equations implies that the solution $u(t,x)$ of
\Eqref{eq:inhomEq} depends as a $C^m$ function on $x$. Hence, if
$H[w]$ rigorously represents a particular solution of
\Eqref{eq:inhomEq}, then i) the function $H[w](t,x)$ is in $C^m(U)$
with respect to $x$ for each $t>0$, ii) we can take the spatial
derivatives under the integral, iii) each spatial derivative of
$H[w](t,x)$ converges from the non-degenerate to the degenerate case at
degenerate points as a continuous function, and vice versa.


\subsubsection*{Behavior near the singular time}

Now we go beyond a formal derivation and determine precise conditions
on $w$ under which the integrals in \Eqref{eq:inhomSol} make sense and
\eqref{eq:inhomSol} provides actual solutions of \Eqref{eq:inhomEq}. We use here the notation 
$\Re$ for the real part of a complex number. 

\begin{proposition}[Pointwise properties of the solution operator $H$]
  \label{prop:behaviorH}
Fix some $\delta>0$, a compact set $K\subset U$, and $l, m \geq 0$, 
and let $w$ be a function in $C^{l\times m}((0,\delta]\times U)$, 
and $a,b \in C^m(U)$.  
In addition, suppose that $w$ satisfies the following asymptotic conditions: 
there exists a constant $\alpha$ such that 
  \[ 
    \alpha>-\Re\lambda_2(x), \qquad  x \in K  
  \]  
and, for all $0\le p\le l$, $0\le q\le m$, 
  \[ 
    \sup_K \big| D^p\del_x^q w(t,\cdot) \big| = O(t^{\alpha}).
  \]
  Then, the operator $H$ given by \Eqref{eq:defH} is well-defined and,
  if $l\ge 2$, it provides a particular classical solution of
  \Eqref{eq:inhomEq}.  Moreover, we have $H[w]\in C^{l\times
    m}((0,\delta]\times K)$ with, for all sufficiently small $\eps>0$,
  $$
\sup_K \big| D^p\del_x^q H[w](t,\cdot) \big| = O(t^{\alpha- \eps}).
  $$
  In addition, for $1
  \leq p \leq l$ one has
  \be
  \label{eq:expressionsH}
  (D_t^{p} H[w])(t,x)=
  \int_0^{t}D_s^{p-1}w(s,x)s^{a(x)-1}
  \left(a(x)\ln\frac st+1\right)\, ds
  \ee
  when $a^2(x)=b(x)$ and, otherwise, 
   \be
  \label{eq:expressionsH2}
  \begin{split}
    (D_t^{p} H[w])(t,x)= 
    \frac{1}{\lambda_1(x)-\lambda_2(x)}\Biggl(
   & -\lambda_2(x)t^{-\lambda_2(x)}
    \int_0^{t}D_s^{p-1} w(s,x)
    s^{\lambda_2(x)-1}ds \\
    &+\lambda_1(x)t^{-\lambda_1(x)}
    \int_0^{t}D_s^{p-1} w(s,x)
    s^{\lambda_1(x)-1}ds \Biggr).
  \end{split}
  \ee 
\end{proposition}

We note that we are allowed to choose $\epsilon=0$ in the previous
proposition only if $a(x)$ and $b(x)$ are constants in space or for
$q=0$. The constant $\epsilon>0$ is necessary in order to control
logarithms which arise when spatial derivatives are taken of functions
involving spatially dependent powers of $t$.

\begin{proof} The proof is quite direct and we only show the
  derivation of the formulas \eqref{eq:expressionsH} and
  \eqref{eq:expressionsH2}. For definiteness, we treat the case
  $a^2(x)\not=b(x)$ only, since the other case is treated
  similarly and the transition between the two cases can be obtained
  by appropriate limiting procedures. Consider the expression for $H$
  in \Eqref{eq:inhomSol}. Under our regularity assumptions, all
  derivatives are calculated by differentiation under the integral
  sign, and we obtain
  \[
  (D^p H[w])(t,x)
  = \frac{1}{\lambda_1-\lambda_2}
  \int_1^{\infty}D_t^p w(t/\zeta,x)
  \left(\zeta^{-\lambda_2-1}-\zeta^{-\lambda_1-1}\right)      
  d\zeta.
  \]
  As done earlier, we introduce the new variable $s:=t/\zeta$ and
  observe that $D_t=t\del_t=s\del_s=:D_s$ for fixed $\zeta$. Hence, we
  obtain
  \begin{align*}
    &(D^p H[w])(t,x)
    = \frac{1}{\lambda_1-\lambda_2}
    \int_0^{t}D_s^p w(s,x)
    \left(\zeta^{-\lambda_2-1}-\zeta^{-\lambda_1-1}\right)\frac{\zeta}s
    ds\\
    =& \frac{1}{\lambda_1-\lambda_2}
    \int_0^{t}\del_s(D_s^{p-1} w)(s,x)
    \left(t^{-\lambda_2}s^{\lambda_2}-t^{-\lambda_1}s^{\lambda_1}\right)
    ds\\
    =& \frac{1}{\lambda_1-\lambda_2}\Biggl(
    t^{-\lambda_2}\left(\left.D_s^{p-1} w(s,x)
    s^{\lambda_2}\right|_0^t-\lambda_2\int_0^{t}D_s^{p-1} w(s,x)
    s^{\lambda_2-1}ds\right)\\
    & \qquad \qquad -t^{-\lambda_1}\left(\left.D_s^{p-1} w(s,x)
    s^{\lambda_1}\right|_0^t-\lambda_1\int_0^{t}D_s^{p-1} w(s,x)
    s^{\lambda_1-1}ds\right)
    \Biggr),
  \end{align*}
  where we used integration by parts. Now, in view of our regularity
  assumptions, we conclude that all terms here have a limit when $t
  \to 0$.  All terms, except for the main integrals, either vanish or
  cancel.
\end{proof}


\subsection{Asymptotic solutions of second-order Fuchsian systems}
\label{sec:asymptoticsolutions}

\subsubsection*{Asymptotic data and canonical expansions}

We now return to the non-linear problem
\Eqref{eq:secondorderFuchsian}.  We are going to identify the
``canonical'' asymptotic behavior (at $t=0$) of general solutions to
second-order Fuchsian equations using first heuristic arguments, and
we determine a ``canonical'' expansion.  Such an expansion involves
certain free functions, interpreted as ``data on the singularity'' and
allows us to formulate a singular initial value problem.

The basic understanding of the behavior of Gowdy solutions at $t=0$ is
obtained from the BKL conjecture, and the idea (which we explain in
further detail in the second paper \cite{BeyerLeFloch2}) is to neglect
spatial derivatives in the evolution equations while solving the
remaining ordinary differential equations at each spatial point
$x$. This approach allows to identify the leading order terms in the
expansion of the solution at $t=0$.  Indeed, the Gowdy equations turn
out to be of second-order Fuchsian form, as shown in the second paper.
This suggests that we use similar arguments for the derivation of
solutions of general second-order Fuchsian equations.

According to the above heuristics, the behavior of general solutions
to \Eqref{eq:secondorderFuchsian} should be driven by the principal
part of the PDE's with coefficients evaluated at $t=0$ if the
source-term satisfies certain ``decay properties'' at $t=0$, as we
will discuss later.  More precisely, the two leading terms in the
expansion of solutions at $t=0$ should be determined by the
homogeneous equations obtained by setting the right-hand side $f$ in
\Eqref{eq:secondorderFuchsian} to zero. According to the discussion
in \Sectionref{sec:linearODEcase}, any solution $u$ to the
second-order Fuchsian equations \Eqref{eq:secondorderFuchsian} should
hence have an expansion of the form
\be
  \label{eq:intuitiveexpansion}
  u^{(i)}(t,x)=
  \begin{cases}
    u_*^{(i)}(x)\,t^{-a^{(i)}(x)}\ln
    t +u_{**}^{(i)}(x)\,t^{-a^{(i)}(x)}+O(t^{-a^{(i)}(x)+\alpha^{(i)}}), & 
 \\   
 & \hskip-2.cm (a^{(i)}(x))^2 = b^{(i)}(x), 
 \\
    & \\
    u_*^{(i)}(x)\,t^{-\lambda_1^{(i)}(x)} 
    +u_{**}^{(i)}(x)\,t^{-\lambda_2^{(i)}(x)}
    +O(t^{-\Re\lambda_2^{(i)}(x)+\alpha^{(i)}}),
    \\
     & \hskip-2.cm  (a^{(i)}(x))^2 \neq b^{(i)}(x),
  \end{cases}    
\ee
for each $i=1,\ldots,n$ and each $x\in U$.  Here, the functions
$u_{*}^{(i)}$ and $u_{**}^{(i)}$ are prescribed, and $\alpha^{(i)}>0$
are real constants. The meaning of the Landau symbols $O$ in this
context will be made precise later; at this stage of the discussion
they have to be understood ``intuitively'' as representing terms of
higher order in $t$ at $t=0$.  In the case of a transition from the
non-degenerate case to the degenerate one or vice versa, the
renormalization given by \Eqref{eq:transitiondata} is necessary and will
always be assumed.

At this stage of the discussion, we clearly see the dependence of
the expected leading-order behavior at $t=0$ on the coefficients of
the principal part of the equation. If the roots $\lambda_1$ and
$\lambda_2$ are real and distinct, i.e.\ if $a^2>b$, we expect a
{\sl power-law} behavior. In the degenerate case $\lambda_1=\lambda_2$,
i.e.\ if $a^2=b$, we expect a {\sl logarithmic} behavior. Finally, when
$\lambda_1$ and $\lambda_2$ are complex for $a^2<b$, the solution is
expected to have an {\sl oscillatory} behavior at $t=0$ of the form
\[u(t,x)=t^{-a(x)}\bigl(\tilde u_*\cos(\lambda_I(x)\log
t)+\tilde u_{**}\sin(\lambda_I(x)\log t)\bigr)+\ldots
\] 
for some real coefficient functions $\tilde u_*(x)$ and $\tilde u_{**}(x)$;
note that in this case, $\lambda_1=\bar\lambda_2=a+i\lambda_I$ with
$\lambda_I:=\sqrt{b^2-a}$.


\subsubsection*{Relevant function spaces}
 
Consider any second-order Fuchsian system of the form described in
\Defref{def:Fuchsian}, with coefficients $a,b, \lambda_1, \lambda_2$
satisfying \Eqref{eq:deflambda2}.  To simplify the presentation, we
restrict attention to scalar equations ($n=1$) and shortly comment on
the general case in the course of the discussion.

Fix some integers $l,m \geq 0$ and constants $\alpha,\delta>0$.  For $w\in
C^{l}((0,\delta],H^m(U))$, we define the norm 
\be
  \label{eq:Banachspace}
  \|w\|_{\delta,\alpha,l,m}:=\sup_{0<t\le \delta} 
  \left(
    \sum_{p=0}^l\sum_{q=0}^m\int_U t^{2(\Re\lambda_2(x)-\alpha)}
    \,\bigl|\del_x^q D^pw(t,x)\bigr|^2 \, dx
  \right)^{1/2}, 
\ee
and denote by $X_{\delta,\alpha,l,m}$ the space of all functions with
finite norm $\|w\|_{\delta,\alpha,l,m}<\infty$.  Throughout, $H^m(U)$
denotes the standard Sobolev space and we recall that all functions
are periodic in the variable $x$ with $U$ being a periodicity domain.
To cover a system of $n\ge 1$ second-order Fuchsian equations, the
norm above is defined by summing over all vector components with
different exponents used for different components; recall that each
equation in the system will have a different root function $\lambda_2$
and we allow that $\alpha=(\alpha^{(1)},\ldots,\alpha^{(n)})$ is a
vector of different positive constants for each equation. The constant
$\delta$, however, is assumed to be common for all equations in the
system. With this modification, all results in the present section
remain valid for systems of equations.

Throughout it is assumed that $\Re\lambda_2$ is continuous and it is 
then easy to check that $(X_{\delta,\alpha,l,m},\|\cdot\|_{\delta,\alpha,l,m})$ is a Banach
space and that the following property holds. 

\begin{lemma}[Approximation by smooth functions]
  \label{lem:density}
  Given any $w\in X_{\delta,\alpha,l,m}$, and constant $\epsilon>0$, 
  there exists a sequence $(w_n)\in
  X_{\delta,\alpha,l,m}\cap C^\infty((0,\delta]\times U)$
  such that
  \begin{equation}
    \label{eq:mollifierconvergence}
    \lim_{n\rightarrow\infty}\|w-w_n\|_{\delta,\alpha-\epsilon,l,m}=0.
  \end{equation}
\end{lemma}

The functions $w_n$ above are taken to be periodic in space, for instance: 
\be
  \label{eq:mollifiers}
  w_\eta(t,x)=t^{-\lambda_2(x)}\int_{-\infty}^\infty\int_0^\infty 
  t^{\lambda_2(y)}w(s,y)k_{\eta}\Bigl(\ln\frac st\Bigr) k_\eta(x-y) 
  \frac 1s dsdy.
\ee
Here, $k_\eta:\R\rightarrow\R_+$ is any smooth kernel supported in
$[-\eta,\eta]$, satisfying $\int_\R k_\eta(x)dx=1$ for all positive
$\eta$. We note that the constant $\epsilon>0$ in
\Eqref{eq:mollifierconvergence} is introduced in order to guarantee
uniform convergence on $(0,\delta]$ in the case where $w$ has
no limit at $t=0$, being just bounded and continuous.

In terms of the spaces $X_{\delta,\alpha,l,m}$, we define in a
mathematically precise way a notion of canonical expansions and
asymptotic data as follows.

\begin{definition}
  \label{def:remainder}
  Consider a second-order Fuchsian equation
  \Eqref{eq:secondorderFuchsian} with continuous coefficients $a, b,
  \lambda_1, \lambda_2$.  Suppose that $v$ and $w$ are functions
  related as follows:
  \be
    \label{eq:splitv}
    v(t,x)=
    \begin{cases}
      u_*(x)\,t^{-a(x)}\ln t+u_{**}(x)\,t^{-a(x)}+w(t,x),
      & \quad (a(x))^2=b(x),\\
      u_*(x)\,t^{-\lambda_1(x)}+u_{**}(x)\,t^{-\lambda_2(x)} 
      +w(t,x),
      & \quad (a(x))^2 \neq b(x),
    \end{cases}    
  \ee
  for some prescribed data $u_*$ and $u_{**} \in H^{m'}(U)$,
   where $m'$ is
  some non-negative integer.  Then, one says that $v$ satisfies a
  \textbf{canonical two-term expansion}
  with \textbf{asymptotic data} $u_*$
  and $u_{**}$ and \textbf{remainder} $w$, provided $w\in
  X_{\delta,\alpha,l,m}$ for some constants $\delta,\alpha>0$ and
  non-negative integers $l, m$.
\end{definition}
If the coefficients of the equations are such that there is a continuous
transition between the two cases in \Eqref{eq:splitv}, then the
asymptotic data functions $u_*$ and $u_{**}$ are renormalized
by \Eqref{eq:transitiondata}.


\subsection{An existence result for second-order Fuchsian 
  ODE systems}

An important property of the ODE solution operator $H$ defined in
\Eqref{eq:defH} is derived now.

\begin{proposition}[Continuity of the ODE solution operator $H$] 
  \label{prop:Hestimate}
  Pick up any constants $\delta>0$, $\alpha>0$ and any integers
  $l\ge 1$, $m\geq 0$.  
  Then, for each $\epsilon>0$, the operator $H$ defined in \Eqref{eq:defH}
  extends uniquely to a continuous linear map
  $X_{\delta,\alpha,l-1,m}\rightarrow X_{\delta,\alpha - \eps,l,m}$,
  and there exists a constant $C_\eps>0$ (independent of $\delta$
  provided $\delta$ is sufficiently small), so that 
  \be
    \label{eq:Hestimate}
    \|H[w]\|_{\delta,\alpha-\eps,l,m}
    \le C_\eps \delta^\epsilon \|w\|_{\delta,\alpha,l-1,m},
  \ee
  for all $w\in X_{\delta,\alpha,l-1,m}$.
\end{proposition}
We stress at this stage that for all $l\ge 2$ the
extended solution operator $H$ indeed provides the general solution to
the Fuchsian equation: given any function $g\in
X_{\delta,\alpha,l-1,m}$, the function $w :=H[g]\in
X_{\delta,\alpha-\eps,l,m}$ satisfies the second-order equation
\[
D^2 w + 2 a \, Dw + b \, w = g,
\]
as equality between functions in the space $X_{\delta,\alpha-\eps,l-2,m}$.

\begin{proof}[Proof of \Propref{prop:Hestimate}]
  We restrict attention to the case $a^2\not=b$, the proof of the
  other case is completely similar and the transition between the two
  cases can be understood as a limiting process.  Consider first a
  function $w\in X_{\delta,\alpha,l-1,m}\cap C^\infty((0,\delta]\times
  U)$. We have
  \[
  \|H[w]\|_{\delta,\alpha-\epsilon,l,m}^2
  =\sup_{0<t\le\delta}
    \int_U t^{2(\Re\lambda_2-\alpha+\epsilon)}\sum_{q=0}^m
    \sum_{p=0}^l|\del_x^q D^pH[w]|^2 \, dx,  
  \]
  which is finite thanks to \Propref{prop:behaviorH} for all
  $\epsilon>0$.  The term $p=0$ is expressed explicitly by means
  of \Eqref{eq:defH}, while the case $p>0$ is treated in
  \Eqref{eq:expressionsH} and \eqref{eq:expressionsH2}.  Suppose
  $(t,x)\in (0,\delta]\times U$. For convenience, we introduce
  functions $w_1$ and $w_2$ so that
  $$
  \del_x^q D^pH[w](x)
  =\int_0^t\del_x^q(w_2(t,x,s)-w_1(t,x,s))s^{-1}ds.
  $$
  For $p=0$, these functions are
  \begin{align*}
    w_1(t,x,s)&:=\frac 1{\lambda_1(x)-\lambda_2(x)}
    w(s,x)t^{-\lambda_1(x)}s^{\lambda_1(x)},\\
    w_2(t,x,s)&:=\frac 1{\lambda_1(x)-\lambda_2(x)}
    w(s,x)t^{-\lambda_2(x)}s^{\lambda_2(x)},
  \end{align*}
  while for $p>0$ we only need to substitute $w$ by $D^{p-1}w$ and
  complement the expression by factors whose particular form is not
  relevant for the following. With this, we get
  \begin{align*}
    |\del_x^q D^pH[w](x)|^2
   & =\left|\int_0^t\del_x^q(
      w_2(t,x,s)-w_1(t,x,s))s^{-1}ds\right|^2\\
    &\le 2\, \left|\int_0^t\del_x^q
        w_2s^{-1}ds\right|^2
      + 2 \, \left|\int_0^t\del_x^q
        w_1s^{-1}ds\right|^2        
  \end{align*}
  The first term (and in the same way the second one) is handled
  via Cauchy-Schwarz's inequality (for some constant $\eta>0$) 
  \begin{align*}
    \int_0^t|\del_x^q w_2s^{-1}|ds
    &=\int_0^t|\del_x^q w_2|s^{-(1+\eta)/2}
    s^{-(1-\eta)/2}ds\\
    &\le\left(\int_0^t|\del_x^q
      w_2|^2s^{-1-\eta}ds\right)^{1/2}
    \left(\int_0^t s^{-1+\eta}ds\right)^{1/2}\\
    &=\left(\frac 1\eta t^\eta\int_0^t|\del_x^q
      w_2|^2s^{-1-\eta}ds\right)^{1/2}.
  \end{align*}
  Hence, we have 
  \begin{equation}
    \label{eq:intermediatecomp1}
    \begin{split}
      &\int_U t^{2(\Re\lambda_2-\alpha+\epsilon)}
      |\del_x^q D^p H[w](x)|^2 dx
      \\
      & \le \frac 2\eta\Bigl(
      \int_U 
      \int_0^tt^{2(\Re\lambda_2-\alpha+\epsilon)}
      t^\eta(\del_x^q
      w_2)^2 s^{-1-\eta}ds dx
      +\int_U 
      \int_0^tt^{2(\Re\lambda_2-\alpha+\epsilon)}
      t^\eta(\del_x^q
      w_1)^2 s^{-1-\eta}ds dx
      \Bigr).
    \end{split}
  \end{equation}
  The first term on the right side of this inequality can be
  written as
  \begin{equation}
    \begin{split}
      \label{eq:intermediatecomp3}
      &\int_U 
      \int_0^tt^{2(\Re\lambda_2-\alpha+\epsilon)}
      t^\eta(\del_x^q
      w_2)^2 s^{-1-\eta}ds dx\\
      &=\int_0^t\left(\int_U 
        \left(s^{2(\Re\lambda_2-\alpha)}
          (\del_x^q w_2 t^{\lambda_2} s^{-\lambda_2})^2
          s^{2(\epsilon_1-\eta)}t^{2(\epsilon-\epsilon_1)}\right)
        \left(\frac st\right)^{2(\alpha-\epsilon_1)+2 i \Im\lambda_2}
        dx\right)
      t^\eta s^{-1+\eta} 
      ds\\
      &\le \frac 1\eta t^{2\eta} \sup_{0<s<t}\left(\int_U 
        \left(s^{2(\Re\lambda_2-\alpha)}
          (\del_x^q w_2 t^{\lambda_2} s^{-\lambda_2})^2
          s^{2(\epsilon_1-\eta)}t^{2(\epsilon-\epsilon_1)}\right)dx\right).
    \end{split}
  \end{equation}
  Here we have assumed that the constant $\epsilon_1$ satisfies
  $\epsilon_1\le\alpha$. The significance of the terms
  $s^{2(\epsilon-\eta)}$ and $t^{2(\epsilon-\epsilon_1)}$ becomes
  clear in a moment.  For the second term on the right side of
  \Eqref{eq:intermediatecomp1}, we get similarly 
  \be
    \label{eq:intermediatecomp2}
    \begin{split}
      &\int_U 
      \int_0^tt^{2(\Re\lambda_2-\alpha+\epsilon)}
      t^\eta(\del_x^q
      w_1)^2 s^{-1-\eta}ds dx\\
      &=\int_0^t\left(\int_U 
        \left(s^{2(\Re\lambda_2-\alpha)}
          (\del_x^q w_1 t^{\lambda_1} s^{-\lambda_1})^2
          s^{2(\epsilon_1-\eta)}t^{2(\epsilon-\epsilon_1)}\right)
        \left(\frac st\right)^{2(\alpha-\epsilon_1)
          +2(\lambda_1-\lambda_2)+2 i \Im\lambda_2}
        dx\right)
      t^\eta s^{-1+\eta} 
      ds\\
      &\le \frac 1\eta t^{2\eta}\sup_{0<s<t}\left(\int_U 
        \left(s^{2(\Re\lambda_2-\alpha)}
          (\del_x^q w_1 t^{\lambda_1} s^{-\lambda_1})^2
          s^{2(\epsilon_1-\eta)}t^{2(\epsilon-\epsilon_1)}\right)dx\right),
    \end{split}
  \ee
  where in the last step we used that $\Re\lambda_1\ge\Re\lambda_2$. Now, for
  $p=0$ we compute,
  \begin{align*}
    w_1 t^{\lambda_1}s^{-\lambda_1}
    &=\frac 1{\lambda_1-\lambda_2}w(s,x),\\
    (\del_x w_1) t^{\lambda_1}s^{-\lambda_1}
    &=\del_x\left(\frac 1{\lambda_1-\lambda_2}w(s,x)\right)
    +\left(\frac 1{\lambda_1-\lambda_2}w(s,x)\right)
    \del_x\lambda_1\, \ln\frac ts,
  \end{align*}
  etc.; analogous formulas hold for $w_2$ and for $p>0$. So the terms
  $\del_x^q w_1 t^{\lambda_1} s^{-\lambda_1}$ incorporate spatial
  derivatives of $w$ of all order lower or equal to $q$, and all
  orders lower than $q$ are multiplied with logarithmic terms in $t$
  and $s$. Hence, in order to guarantee that the suprema in
  \Eqref{eq:intermediatecomp3} and \eqref{eq:intermediatecomp2} are
  finite for each $t>0$, we must choose $\epsilon_1>\eta$. Moreover,
  the suprema are uniformly bounded for all $t\in (0,\delta]$, if
  $\epsilon>\epsilon_1$.  With these choices, we can rearrange all
  terms, introduce a finite constant $C$, which is independent of
  $\delta$ if $\delta$ is small, as in the hypothesis, and hence
  obtain
  \[ 
    \|H[w]\|_{\delta,\alpha-\eps,l,m}
    \le C \delta^\eta \|w\|_{\delta,\alpha,l-1,m},
  \]
  for all $w\in X_{\delta,\alpha,l-1,m}\cap C^\infty((0,\delta]\times
  U)$. Since this inequality holds for all $\eta$ smaller than
  $\epsilon$, it also holds in the limit
  $\eta\rightarrow\epsilon$. The constant $C$ can hence be adapted so
  that \Eqref{eq:Hestimate} follows for all $w\in
  X_{\delta,\alpha,l-1,m}\cap C^\infty((0,\delta]\times U)$.

  Now let $w$ be a general element in $X_{\delta,\alpha,l-1,m}$
  and
  choose a positive $\epsilon_0$ with $\epsilon_0<\epsilon$. We set
  $\widetilde\epsilon:=\epsilon-\epsilon_0$. According to
  \Lemref{lem:density}, there exists a sequence $(w_n)\subset
  X_{\delta,\alpha,l-1,m}\cap C^\infty((0,\delta]\times U)$, so that
  $\lim_{n\rightarrow\infty}\|w-w_n\|_{\delta,\alpha-\epsilon_0,l-1,m}=0$. Our
  results for the smooth case
   show that $(H[w_n])$ is a Cauchy
  sequence in $X_{\delta,\alpha-\epsilon_0-\widetilde\epsilon,l-1,m}
  =X_{\delta,\alpha-\epsilon,l-1,m}$, and we can denote its limit
  element by $H[w]\in X_{\delta,\alpha-\epsilon,l-1,m}$. In this way,
  we define the extension of $H$ to the whole space
  $X_{\delta,\alpha,l-1,m}$. It is then straightforward to see that
  the limit of the estimate \Eqref{eq:Hestimate} leads to the claimed
  estimate for the full space $X_{\delta,\alpha-\epsilon_0,l-1,m}$.
\end{proof}

Consider a Fuchsian equation \Eqref{eq:secondorderFuchsian} with
right-hand side $f$ of the form \Eqref{eq:inhomog}. Let $v$ be a
function in the form of \Defref{def:remainder} with remainder $w \in
X_{\delta,\alpha,l,m}$ and prescribed data $u_*,u_{**}\in H^{m'}(U)$
for some $\delta>0$, $\alpha>0$ and $l,m,m'\geq 0$.  Suppose that $m,
m'$ are sufficiently large.
If necessary, the derivatives in \Eqref{eq:inhomog} are understood in
the sense of distributions
and we set
\be
  \label{eq:ddefoperatorF}
  F[w](t,x):=f[v](t, x)
\ee
and finally
\begin{equation}
  \label{eq:operatorG}
  G:=H\circ F.
\end{equation}
We are now in a position to define an iteration
sequence based on this operator $G$.

\begin{proposition}[Iteration sequence]
  \label{prop:iterationsequence}
  With the same notation and assumptions as in
  \Propref{prop:Hestimate}, let $k$ be the number of spatial
  derivatives in $f$ according to \Eqref{eq:inhomog} and $m_0$ another
  non-negative integer.  Suppose that, for given asymptotic data, the
  operator $F$ satisfies the following {\it regularity assumption}
  (for some $\eps_0>0$)
  \begin{equation}
    \label{eq:consistentF}
    F:X_{\delta,\alpha,l,m}\rightarrow
    X_{\delta,\alpha+\eps_0,l-1,m-k},
  \end{equation}
  for an integer $l\ge 1$, and all non-negative integers $m$ with
  $k\le m\le m_0$. Let $G$ be the operator defined in
  \Eqref{eq:operatorG}. Then, given any $w_1\in
  X_{\delta,\alpha,l,m_0}$, the (in general finite) sequence $(w_j)$
  determined by
  \begin{equation}
    \label{eq:iteration}
    w_{j+1}=G[w_j],     \qquad\text{for all integers } j\in[1,m_0/k]
  \end{equation}  
  is well-defined and, moreover, $w_{j+1}\in
  X_{\delta,\alpha,l,m_0-j k}$.
\end{proposition}

Consider a sequence $(w_j)$ defined in the above lemma. It determines
a sequence $(v_j)$ for fixed asymptotic data according to
\Eqref{eq:splitv}, and all $v_j$ satisfy the canonical two-term
expansion.  
For $l, j \ge 2$, the function $w_j\in X_{\delta,\alpha,l,m}$
satisfies the second-order equation
\begin{equation*} 
  D^2 w_j+2a \, D w_j+b \, w_j
  = F[w_{j-1}]
\end{equation*}
as equality in the space $X_{\delta,\alpha,l-2,m}$. The sequence
$(w_j)$ has infinitely many elements if $m_0=\infty$ or if $k=0$. In
the latter case the second-order Fuchsian equation is a system of
ordinary differential equations (with the spatial variable $x$ as a
parameter). In both cases, a function $w$ is a fixed point of the
iteration sequence if and only if the associated function $v$ is a
solution of \Eqref{eq:secondorderFuchsian}.  A fixed point theorem for
the ODE case is as follows.

\begin{theorem}[Existence of solutions to second-order Fuchsian ODEs]
  \label{th:fixedpointtheorem}
  Under the assumptions as in \Propref{prop:iterationsequence} with $k=0$, 
  suppose additionally that for given asymptotic data, the operator $F$ satisfies the following {\bf Lipschitz continuity
  property}: for each $r>0$ and 
  $\eps_0>0$ arising in \Eqref{eq:consistentF}, there exists $\widehat C >0$
  independent of $\delta$, so that
  \be
    \label{eq:ConvergenceCondF}
    \|F[w]-F[\widetilde w]\|_{\delta,\alpha+\epsilon_0,l-1,m} \le \widehat C \, 
    \|w-\widetilde w\|_{\delta,\alpha,l,m}
  \ee
  for all $w,\widetilde w\in \overline{B_r(0)}\subset
  X_{\delta,\alpha,l,m}$.  Then, given any initial data $w_1\in
  X_{\delta,\alpha,l,m}$ and provided $\delta>0$ is sufficiently
  small, the iteration sequence \Eqref{eq:iteration} converges to a
  unique fixed point $w\in X_{\delta,\alpha,l,m}$.
\end{theorem}

\begin{proof}
  Our previous results imply
  \[\|G[w]-G[\widetilde w]\|_{\delta,\alpha,l,m} \le \widetilde C \delta^\eta
  \|w-\widetilde w\|_{\delta,\alpha,l,m}
  \]
  for a constant $\widetilde C>0$, provided $w,\widetilde w\in
  \overline{B_K}(0)\subset X_{\delta,\alpha,l,m}$. Hence, for
  sufficiently small $\delta$, the operator $G$ becomes a
  contraction. The convergence of the iteration sequence follows if we
  can guarantee that $w_j\in \overline{B_K(0)}$ for a
  sufficiently large $K$. This, however, is the case since the
  sequence $(w_j)$ is a Cauchy sequence thanks to the contraction
  property of $G$ and, hence, is bounded.
\end{proof}

Condition \Eqref{eq:ConvergenceCondF} guarantees convergence of the
sequence $(w_j)$ in the ODE case. The approach of this section is not
sufficient to cover the PDE cases of interest and, in general, 
this condition does not hold if $k>0$. Moreover, the iteration
sequence constructed above is only finite if $m_0$ is finite.  We can
expect that in typical applications, $m_0$ is infinite if the
asymptotic data are smooth and, say, $w_1=0$. Still, this does not lead
to an existence result except for the analytic case, see
\cite{KichenassamyRendall}. Well-posedness for second-order Fuchsian
PDE's in a larger than the analytic class will be addressed
in \Sectionref{sec:hyperbolicfuchsianequations} after we make the
additional assumption that the Fuchsian equations are hyperbolic.


\subsection{Asymptotic solutions of arbitrary order}
\label{sec:asymptsol}
The iterative sequence $(w_j)$ has useful asymptotic properties.  In
order to simplify the discussion in this section, we assume, instead
of \Eqref{eq:inhomog}, that $f$ has the form
\be
  \label{eq:inhomog2}
  f[u](t,x):=f(t, x; u, Du, \del_x u, \del_x Du,\del_x^2 u) 
\ee 
and that it is a polynomial in all of the arguments involving $u$ with
coefficients which are smooth and spatially periodic on
$(0,\delta]\times U$.  The operator $F$ is defined, in the same way as
was done earlier,
 from given asymptotic data $u_*$ and $u_{**}$. We henceforth
assume in the following that $u_*,u_{**}\in H^{m_1}(U)$ for some
non-negative integer $m_1$ and that \Eqref{eq:consistentF} holds for
$k=2$ and $m_0=m_1$.

\begin{definition}
  A function $v$ satisfying the canonical two-term expansion with
  given asymptotic data is called an {\bf asymptotic solution of order
    $\gamma > 0$} to the system \Eqref{eq:secondorderFuchsian}
  provided the \textbf{residual}
  \[
  R[w]:=L[w]-F[w]
  \] 
belongs to $X_{\delta,\gamma+\Re\lambda_2,0,0}$, in which 
the following notation
  \be
    \label{eq:defL}
    L[w] := D^2 w + 2a \, Dw + b \, w,
  \ee
is used for the principal part of \Eqref{eq:secondorderFuchsian} and one 
assumes 
  that $w\in X_{\delta,\alpha,l,m}$ with $l,m\ge 2$.
\end{definition}

In this definition, we use the obvious generalization of the spaces
$X_{\delta,\widetilde\alpha,l,m}$ to spatially dependent exponents
$\widetilde\alpha$. Note that $L[w]=L[v]$ if $w$ and $v$ are related
as in \Eqref{eq:splitv}.  When $l=1$ or $m=1$, one needs to reformulate
the operator $L$ (hence $R$) in a weak form, as we will explain below.  

\begin{theorem}
  \label{th:orderofsequence}
  Suppose that the operator $F$ satisfies the conditions stated earlier for given asymptotic data, and
  consider the iteration sequence $w_j\in X_{\delta,\alpha,l,m_1-2(j-1)}$
  for $j\ge 2$ given by \Eqref{eq:iteration} with $w_1=0$.  Then, 
  for any constant $\kappa<1$ the sequence has the property
  \[
  w_{j+1}-w_j\in X_{\delta,\alpha+(j-1)\kappa\epsilon_0,l,m_1-2j} 
  \] 
  for $1\le j\le m_1/2$. Moreover, the residual
  satisfies
  \[
  R[w_{j}]\in X_{\delta,
    \alpha+(j-1)\kappa\epsilon_0,l-1,m_1-2(j-2)},
  \]
  and hence, $w_j$ is an asymptotic solution of
  order $\gamma=-\Re\lambda_2+\alpha+(j-1)\kappa\epsilon_0$
  for $2\le j$. 
\end{theorem}

This establishes that the order of the asymptotic solution $w_j$ is an
increasing function in $j$.  Hence, the functions $w_j$ can be
interpreted as approximations of actual solutions of increasing
accuracy at $t=0$.
 
\begin{proof} 
  In order to show the first relation, we proceed inductively and
  start with $j=1$. We need to show that $w_2\in
  X_{\delta,\alpha,l,m_1-2}$ which is true by
  \Propref{prop:iterationsequence}.  Next, we suppose that
  \[
  w_{j}-w_{j-1}\in X_{\delta,\alpha+(j-2)\kappa\epsilon_0,l,m_1-2(j-1)},
  \]
  has already been shown for a given integer $j\in[2,m_1/2+1]$.  We
  take the difference of the equations for $w_{j+1}$ and $w_{j}$, and
  the linearity of the operator $H$ implies
  \begin{equation}
    \label{eq:diffsequence}
    w_{j+1}-w_j=H[F[w_j]-F[w_{j-1}]].
  \end{equation}
  Since $f$ of the form \Eqref{eq:inhomog2} depends smoothly on its
  arguments by assumption, the mean value theorem
  implies the existence of a matrix-valued function $M$ of the form
  \begin{align*}
    M[w_j,w_{j-1}](t,x)
    :=M(t, x;\, &w_{j}, Dw_{j}, \del_x w_{j}, \del_x
    Dw_{j},\del_x^2 w_{j},\\
    &w_{j-1}, Dw_{j-1}, \del_x w_{j-1}, \del_x
    Dw_{j-1},\del_x^2 w_{j-1})
  \end{align*} 
  depending as a polynomial on all arguments involving $w_j$ and
  $w_{j-1}$ with the property that first
  \[M[w_j,w_{j-1}]\in 
  X_{\delta,\Re\lambda_2-\alpha+\epsilon_0,l-1,m_1-2(j-2)}, 
  \]
  and that second
  \be
    \label{eq:splitDiffF}
    F[w_j]-F[w_{j-1}]=M[w_j,w_{j-1}]\cdot 
      \Delta V[w_j,w_{j-1}].
  \ee
  Here, the
  vector-valued function $\Delta V$ is defined as
  \begin{align*}
    \Delta V[w_j,w_{j-1}]:=(&w_{j}-w_{j-1}, Dw_{j}-Dw_{j-1}, \del_x
    w_{j}-\del_x w_{j-1},
    \\
    &\del_x Dw_{j}-\del_x Dw_{j-1},\ldots,
    \del_x^2 w_{j}-\del_x^2 w_{j-1})^T.
  \end{align*}
  We have,
  \[\Delta V[w_j,w_{j-1}]\in X_{\delta,
    \alpha+(j-2)\kappa\epsilon_0,l-1,m_1-2(j-2)}.\]
  All this implies that
  \[F[w_j]-F[w_{j-1}]\in X_{\delta,
    \alpha+(j-2)\kappa\epsilon_0+\epsilon_0,l-1,m_1-2(j-2)}.
  \] 
  Then, \Propref{prop:Hestimate} yields
  \[w_{j+1}-w_{j}\in X_{\delta,
    \alpha+(j-1)\kappa\epsilon_0,l,m_1-k(j-2)}.
  \]   
  At this point we see the significance of the requirement
  $\kappa<1$. Namely, we must choose the constant
  $\epsilon$ in \Propref{prop:Hestimate} as
  $\epsilon=\epsilon_0(1-\kappa)$.
  
  The second claim of the theorem is now an immediate consequence. We
  write the system, which determines $w_j$, in the form
  \[L[w_j]-F[w_j]=F[w_{j-1}]-F[w_j].
  \] 
  Again, we can write the right side as above.  We conclude from the
  previous results that the right side is in $X_{\delta,
    \alpha+(j-2)\kappa\epsilon_0+\epsilon_0,l-1,m_1-2(j-2)}$ for $2\le
  j$.  Since the left side equals $R[w_{j}]$, the result follows.
\end{proof}


\section{Second-order hyperbolic Fuchsian systems}
\label{sec:hyperbolicfuchsianequations}

\subsection{Assumptions and basic definitions}
\label{sec:assumptions_SIVP}

From now on we focus on \textit{hyperbolic} second-order Fuchsian
equations in the sense of Definition~\ref{def:secondorderFuchsianHyp} below
---which form a special case of general second-order Fuchsian
equations. Our aim of this section is to establish a well-posedness theory for the
(singular) initial value problem when data are prescribed on the singularity. 

\begin{definition}[Second-order hyperbolic Fuchsian systems]
\label{def:secondorderFuchsianHyp}
A second-order hyperbolic Fuchsian system is a set of partial
differential equations of the form
\be
    \label{eq:secondorderFuchsianHyp}
      D^2 v+2A\, D v+B\, v 
      -t^2K^2\partial_x^2 v
      = f[v],
\ee
in which the function $v:(0,\delta]\times U\rightarrow \R^n$ is the
main unknown (defined for some $\delta>0$ and some interval $U$),
while the coefficients $A=A(x)$, $B=B(x)$, $K=K(t,x)$ are diagonal
$n\times n$ matrix-valued maps and are smooth in $x \in U$ and $t$ in
the half-open interval $(0,\delta]$, and $f=f[v](t, x)$ is an
$n$-vector-valued map of the following form
\be
  \label{eq:hypinh}
  f[v](t,x):=f\Big(t, x, v(t,x), Dv(t,x), t K(t,x) \del_x v(t,x)\Big).
\ee 
\end{definition}

As earlier, we assume that all functions are periodic with respect to
$x$ and that $U$ is a periodicity domain. Further restrictions on the
coefficients and on the right-hand side will be imposed and discussed
in the course of our investigation.  Hence, hyperbolic second-order
Fuchsian systems are second-order Fuchsian systems with a particular
structure of their right-hand side: in our new notation, we have
separated the second-order spatial derivatives from other terms in the
right-hand side $f$ and incorporate them into the principal part of
the equation, which now reads
\be
  \label{eq:defLPDE}
  L:=D^2+2A \, D +B-t^2K^2\partial_x^2.
\ee
This is a linear wave operator for $t>0$ and, indeed,
\Eqref{eq:secondorderFuchsianHyp} is hyperbolic provided $K$ satisfies
(positivity) conditions given below. Later on we will construct solutions
where  the first three terms of the principal part are of
the same order at $t=0$ and dominant as in the previous section, while
the second spatial derivative term is assumed to be of higher order at
$t=0$. Hence, we expect the same phenomenology at $t=0$ as earlier, and
the only significance of the new term in the principal part is that it
allows to derive energy estimates.

The eigenvalues of the matrix $K$ are denoted by $k^{(i)}$ and, in the
scalar case (or when there is no need to specify the index), we
simply write $k$. These quantities are interpreted as
characteristic speeds. Throughout this section, we assume that 
they have the form 
\begin{equation}
  \label{eq:behaviorofk}
  \begin{split}
    &k^{(i)}(t,x)=t^{\beta^{(i)}(x)}\nu^{(i)}(t,x),\\
    &\text{with }\beta^{(i)}:U\rightarrow (-1,\infty),\,
    \nu^{(i)}:[0,\delta]\times U\rightarrow (0,\infty)\text{ smooth functions.}
  \end{split}
\end{equation}
In particular, we assume that each derivative of $\nu^{(i)}$ has a
finite limit at $t=0$ for each $x\in U$. The motivation for these
assumptions will become clear in the forthcoming discussion.  Note
we allow for the characteristic speeds to diverge at
$t=0$. At a first glance, this appears to conflict with the standard
finite domain of dependence property of hyperbolic equations. Recall
that for the standard initial value problem of hyperbolic equations,
the solution at a given point is determined by the restriction of the
data to a bounded domain of the initial hypersurface; this is a
consequence of the finiteness of the characteristic speeds.  A closer
look at the requirement $\beta(x)>-1$, however, indicates that the
characteristic curves are {\sl integrable} at $t=0$ and, hence that the
{\sl finite} domain of dependence property is preserved under our
assumptions.

In order to simplify the presentation and without genuine loss of generality,
we restrict the discussion now to the scalar case $n=1$. 
Consider any second-order Fuchsian hyperbolic equation and 
for each non-negative integer $l$ and real numbers
$\delta,\alpha>0$, define the space $X_{\delta,\alpha,l}$
by
\[X_{\delta,\alpha,l}:=\bigcap_{p=0}^lX_{\delta,\alpha,p,l-p},
\]
and introduce the norm
\[\|f\|_{\delta,\alpha,l}:=\Biggl(\sum_{p=0}^l
  \|f\|_{\delta,\alpha,p,l-p}^2\Biggr)^{1/2}, \qquad 
  f\in X_{\delta,\alpha,l}.
\]
 Recall that the spaces
$X_{\delta,\alpha,l,m}$ and norms $\|\cdot\|_{\delta,\alpha,l,m}$ have
been introduced in the previous section.  It is straightforward to see
that the spaces $(X_{\delta,\alpha,l}, \|\cdot\|_{\delta,\alpha,l})$
have similar properties as the previously defined ones.  As we will
see in the following discussion, it is not possible to control
solutions of our equations in the spaces $X_{\delta,\alpha,l}$
directly. It turns out that we must use spaces $(\tilde
X_{\delta,\alpha,l}, \|\cdot\|_{\delta,\alpha,l}^\sim)$ instead which
are defined as earlier, but in the norm $\|f\|_{\delta,\alpha,l}^\sim$
of some function $f$, the highest spatial derivative term
$\partial^l_x f$ is weighted with the factor $t^{\beta+1}$. Here
$\beta$ is the characteristic speed of the equation given by
\Eqref{eq:behaviorofk}.  It is easy to see under the earlier conditions
that $(\tilde X_{\delta,\alpha,l}, \|\cdot\|_{\delta,\alpha,l}^\sim)$
are Banach spaces. Moreover, for any $w\in \tilde
X_{\delta,\alpha,l}$, the mollified function $w_\eta$ defined by
\Eqref{eq:mollifiers} is an element of $\tilde
X_{\delta,\alpha-\epsilon,l}\cap C^\infty((0,\delta]\times U)$ for
every $\epsilon>0$. Furthermore, the sequence of mollified functions
$w_\eta$ in the limit $\eta\rightarrow 0$ converges to $w$ in the norm
$\|\cdot\|_{\delta,\alpha-\epsilon,l}^\sim$. We also note that
$X_{\delta,\alpha,l}\subset\tilde X_{\delta,\alpha,l}$. For our later
discussion, let us define
\[X_{\delta,\alpha,\infty}:=\bigcap_{l=0}^\infty
X_{\delta,\alpha,l},
\] 
and note that $X_{\delta,\alpha,\infty}=\bigcap_{l=0}^\infty \tilde
X_{\delta,\alpha,l}$.

For $w\in \tilde X_{\delta,\alpha,1}$, the operator $L$ in \Eqref{eq:defLPDE}
is defined in the sense of distributions, only, via the following weak form:  
\begin{equation}
  \label{eq:weakL}
\aligned
   & \langle  \mathcal L[w],\phi\rangle
    \\
    &:=\int_{0}^\delta\int_\R
    t^{\Re\lambda_2(x)-\alpha} \Big(
    && -Dw(t,x) D\phi(t,x)
    +(2A(x)-\Re\lambda_2(x)+\alpha-1)\,Dw(t,x)\phi(t,x)\\
  &  &&+B(x)\,w(t,x)\phi(t,x)  
    + t K(t,x)\partial_x w(t,x) tK(t,x)\partial_x\phi(t,x)\\
  &  &&+(2t\partial_xK(t,x)+\partial_x\Re\lambda_2(x) K(t,x)t\ln t)
    t K(t,x)\partial_xw(t,x)\,\phi(t,x)\Big)\,\,dx dt,
\endaligned
\end{equation}
where $\phi$ is any test function, i.e.~a real-valued
$C^\infty$-function on $(0,\delta]\times\R$ together with some $T\in
(0,\delta)$ and a compact set $K\in\R$ so that $\phi(t,x)=0$ for all
$t>T$ and $x\not\in K$, and each derivative of $\phi$ has a finite
(not necessarily vanishing) limit at $t=0$ for every $x\in U$. The
formula \Eqref{eq:weakL} is obtained as follows. First we assume that
$w$ is a smooth function for $t>0$ in $\tilde X_{\delta,\alpha,1}$. We
compute $L[w]$ and multiply \Eqref{eq:defLPDE} with
$t^{\Re\lambda_2-\alpha}$. Then we integrate this expression in $x$ on
$U$ and in $t$ on $[\epsilon,\delta]$ for some $\epsilon>0$, and
integrate by parts. The resulting expression, which resembles
\Eqref{eq:weakL} plus a boundary term at $t=\epsilon$, is meaningful
for general $w\in \tilde X_{\delta,\alpha,1}$ and only the limit
$\epsilon\rightarrow 0$ remains to be checked. It turns out that the
assumption $w\in \tilde X_{\delta,\alpha,1}$ is sufficient to
guarantee that this limit is finite and in particular that the
boundary term at $t=\epsilon$ resulting from integration by parts goes
to zero in the limit.  For our later discussion, we note that for any
given test function $\phi$, the linear functional $\langle\mathcal
L[\cdot],\phi\rangle: \tilde X_{\delta,\alpha,1}\rightarrow\R$ is continuous
with respect to the norm $\|\cdot\|_{\delta,\alpha,1}^\sim$. This is the
main reason to include the factor $t^{\Re\lambda_2(x)-\alpha}$ in the
definition of $\mathcal L$.

Consider now functions $u$, $v$, $w$ on $(0,\delta]\times U$ related by 
\be
  \label{eq:splitv2}
  v(t,x)=u(t,x)+w(t,x).  
\ee 
In the following, $v$ will stand for a solution of a Fuchsian system.
The function $u$ will be called the {leading-order part} and $w$
the {remainder} of the solution at the singularity at $t=0$. In
agreement with the discussion of the canonical two-term expansion in
\Defref{def:remainder}, we will look for remainders $w$ in spaces
$\tilde X_{\delta,\alpha,l}$ for some $\alpha>0$. However, at this stage of
the discussion we will not yet fix the particular form of the function
$u$ and its dependence on the asymptotic data.  Indeed let us assume that
 $u$ is some given function. In analogy to our earlier discussion,
we introduce the operator $F$ as
\[
F[w](t,x):=f[u+w](t,x).
\]
If the operator $F$ is a map $\tilde X_{\delta,\alpha,1}\rightarrow
\tilde X_{\delta,\alpha,0}, w\mapsto F[w]$, which we shall assume later on,
it is meaningful to define its weak form by (for all test functions $\phi$)
\[\langle\mathcal F[w],\phi\rangle:=
\int_{0}^\delta\int_\R t^{\Re\lambda_2(x)-\alpha}
F[w](t,x)\phi(t,x)dxdt. 
\]

\begin{definition}[Weak solutions of second-order hyperbolic Fuchsian
  systems]
  \label{def:weaksolution}
  Let $u$ be a given function and $\delta, \alpha>0$ be constants. 
  Then, one says 
  that $w\in \tilde X_{\delta,\alpha,1}$ is a weak solution to the second-order
  hyperbolic Fuchsian equation \Eqref{eq:secondorderFuchsianHypForw},
  provided
  \be
    \label{eq:weakequation}    
    \mathcal P[w]:=
    \mathcal L[w]+\mathcal L[u]-\mathcal F[w]
    =0.
  \ee
\end{definition}
For our later discussion, we note that for any given test function
$\phi$ and under our earlier assumptions, 
the linear functional $\langle\mathcal P[\cdot],\phi\rangle$
on $\tilde X_{\delta,\alpha,1}$ is continuous with respect to the norm
$\|\cdot\|_{\delta,\alpha,1}^\sim$.  For
completeness, we also state here that the classical form of
the equation for $w$ \be
  \label{eq:secondorderFuchsianHypForw}
  L[w] = F[w]-L[u], 
\ee 
if $v$ is a classical solution of the equation
\Eqref{eq:secondorderFuchsianHyp}. Clearly, if $w$ and $u$
are sufficiently smooth, then $w$ is a weak
solution if and only if $w$ is a classical solution of
\Eqref{eq:secondorderFuchsianHypForw}, or equivalently if $v$ given by
\Eqref{eq:splitv2} is a classical solution to the original
second-order hyperbolic Fuchsian equation
\Eqref{eq:secondorderFuchsianHyp}.

We arrive at the following important notion. 

\begin{definition}[Singular initial value problem (SIVP)]
  \label{def:SIVP}
  Consider a second-order hyperbolic Fuchsian equation
  \Eqref{eq:secondorderFuchsianHypForw} with coefficients ($a, b,
  \lambda_1, \lambda_2)$ and characteristic speeds $k$ satisfying all the conditions 
  stated earlier. 
  Moreover, choose a leading-order part $u$.  Then, a function
  $v:(0,\delta]\times U\rightarrow\R$ is called a solution of the {\bf
    singular initial value problem} provided $w:=v-u$ belongs
  to $\tilde X_{\delta,\alpha,1}$ for some $\alpha>0$, and is a weak
  solution to the second-order hyperbolic Fuchsian system
  \Eqref{eq:weakequation}.
\end{definition}
In particular, we will be interested in the case when $u$ is
parametrized by asymptotic data in analogy to the canonical two-term
expansion. At this stage of the discussion, however, the particular
form of $u$ is not fixed yet.


\subsection{Linear theory in the space 
\texorpdfstring{$\tilde X_{\delta,\alpha,1}$}
{X{delta,alpha,1}}. Main statement}
\label{sec:Linear1Statement}

In this subsection and the following one, we study a particularly fundamental
case described by the two conditions: 
\begin{enumerate}
\item Vanishing leading-order part: $u\equiv 0$.
\item Linear source-term:
  \begin{equation}
    \label{eq:linearinhom2}
    F[w](t,x)=f_0(t,x)+f_1(t,x) w+f_2(t,x) Dw+f_3(t,x) t k \partial_x
    w,
  \end{equation}  
  with given functions $f_0$, $f_1$, $f_2$, $f_3$, so that
  $f_1$,
  $f_2$, $f_3$ are smooth spatially periodic on $(0,\delta]\times U$, and
  near $t=0$
  \begin{equation}
    \label{eq:linearinhomDecay}
    \sup_{x\in\bar U}f_a(t,x)=O(t^\mu), \qquad a=1,2,3,
  \end{equation}
  for some constant $\mu>0$. 
\end{enumerate} 
We have not made any assumptions for the function $f_0$ yet, since
in the following discussion this function will play a different role than $f_1$, $f_2$, $f_3$. 
Moreover, no loss of generality is implied
by the condition $u\equiv 0$, since
the general case can be recovered by absorbing $L[u]$ into the function $f_0$.

Under these assumptions, we pose the question whether there exists a
unique weak solution $w$ of the given second-order hyperbolic equation
in $\tilde X_{\delta,\alpha,1}$ for some $\delta,\alpha>0$.

\begin{proposition}[Existence of solutions of the linear singular
  initial value problem in $\tilde X_{\delta,\alpha,1}$]
  \label{prop:existenceSVIPlinear}
  Under the assumptions made so far, there
  exists a unique solution $w\in \tilde X_{\delta,\alpha,1}$ of the 
  singular initial value problem for given $\delta,\alpha>0$ provided:
  \begin{enumerate}
  \item The matrix
    \begin{equation}
      \label{eq:defNfinal3}
      N:=
      \begin{pmatrix}
        \Re(\lambda_1-\lambda_2)+\alpha & ((\Im\lambda_1)^2/\eta-\eta)/2 & 0 \\
        ((\Im\lambda_1)^2/\eta-\eta)/2 
        & \alpha 
        & t\partial_x{k}-\partial_x\Re(\lambda_1-\lambda_2)(t k\ln t) \\
        0 & t\partial_x{k}-\partial_x\Re(\lambda_1-\lambda_2)(t k\ln t) 
        & \Re(\lambda_1-\lambda_2)+\alpha-1-Dk/k
      \end{pmatrix}
    \end{equation}
    is positive semidefinite at each
    $(t,x)\in (0,\delta)\times U$ for a constant $\eta>0$.
  \item The source-term function $f_0$ is in
    $X_{\delta,\alpha+\epsilon,0}$ for some $\epsilon>0$.
  \end{enumerate}
  Then, the solution operator
  \[\mathbb H: X_{\delta,\alpha+\epsilon,0}
  \rightarrow \tilde X_{\delta,\alpha,1},\quad f_0\mapsto w,
  \] 
  is
  continuous and there exists a finite constant $C_\epsilon>0$ so that
  \begin{equation}
    \label{eq:continuityHPDE}
    \|\mathbb H[f_0]\|_{\delta,\alpha,1}^\sim \le \delta^{\epsilon}
    C_\epsilon \|f_0\|_{\delta,\alpha+\epsilon,0},
  \end{equation}
  for all $f_0$.
  The constant $C_\epsilon$ can depend on $\delta$, but is bounded for
  all small $\delta$.
\end{proposition}
For reasons that will become clear later on, we call $N$ the
\textbf{energy dissipation matrix}. We have assumed that $\alpha$ is a
positive constant. If, however, $\alpha$ is a positive spatially
periodic function in $C^1(U)$, the definition of the spaces
$X_{\delta,\alpha,k}$ and $\tilde X_{\delta,\alpha,k}$ remains the
same, and only the $(2,3)$- and $(3,2)$-components of the energy
dissipation matrix $N$ change to
$t\partial_x{k}-\partial_x(\Re(\lambda_1-\lambda_2)+\alpha)(t k\ln
t)$. In the following, we continue to assume that $\alpha$ is a
constant in order to keep the presentation as simple as possible, but
we stress that all following results hold (with this slight change of
$N$) if $\alpha$ is a function, and hence no new difficulty arise.

\subsection{Linear theory in the space \texorpdfstring{$\tilde
    X_{\delta, \alpha,1}$}{X{delta,alpha,1}}. The proof}

The main idea for the proof is to approximate a solution of the {\sl
  singular} initial value problem by a sequence of solutions of {\sl
  regular} initial value problems.

\begin{definition}[Regular initial value problem (RIVP)]
  Fix $t_0\in (0,\delta]$ and some smooth periodic functions
  $g,h:U\rightarrow\R$, and suppose that the right-hand side is
  of the form \Eqref{eq:linearinhom2} with given smooth spatially
  periodic functions $f_0$, $f_1$, $f_2$, $f_3$ on $[t_0,\delta]\times
  U$.  Then, $w:[t_0,\delta]\times U\rightarrow\R$ is called a
  solution of the {\bf regular initial value problem}
  associated with the \textbf{regular data} $g, h$ if
  \Eqref{eq:secondorderFuchsianHypForw} holds everywhere on
  $(t_0,\delta]\times U$ and, moreover,
  \[w(t_0,x)=g(x), \quad \partial_t w(t_0,x)=h(x).\]
\end{definition}

For the regular initial value problem, we indeed assume that
$f_0$ is smooth just as $f_1$, $f_2$ and $f_3$. By the general theory of linear hyperbolic
 equations, the regular initial value problem is well-posed, in the sense that there exists a
unique smooth solution $w$ defined on $[t_0,\delta]$ for any choice of smooth initial data.

In order to simplify the presentation, we restrict to the scalar case
$n=1$ for this whole section; the general case can be obtained with
the same ideas.
Choose $\delta,\alpha>0$ and let $w\in C^{1}((0,\delta]\times U)$
be a spatially periodic function. Then, we define its {\bf energy} at
the time $t\in (0,\delta]$ by
\be
  \label{eq:defenergy}
\aligned
    E[w](t):=&  e^{-\kappa t^\gamma} 
    \int_U t^{2(\lambda_2(x)-\alpha)} \, e[w](t,x) \, dx,
    \\
    e[w](t,x) :=&  
    \frac 12\left((\eta\, w(t,x))^2+(Dw(t,x))^2 
      +(t k(t,x)\partial_x w(t,x))^2\right),
    \endaligned
\ee 
for some constants $\kappa\ge 0$, $\gamma>0$ and $\eta>0$.  
For convenience, we also introduce the following notation. For any
scalar-valued function $w$, we define the vector-valued function
\be
  \label{eq:hatw}
  \widehat w(t,x):=t^{\Re\lambda_2(x)-\alpha}
  (\eta w(t,x),Dw(t,x),t k(t,x)\partial_x w(t,x)),
\ee
involving the same constants as in the energy.
Then, we can write
\be
  \label{eq:relationEnergyNorm}
  E[w](t)=\frac 12 e^{-\kappa t^\gamma}\|\widehat
  w(t,\cdot)\|^2_{L^2(U)},
\ee
the norm here being the Euclidean $L^2$-norm for vector-valued
functions in $x$.  It is important to realize that, provided $\eta>0$,
the expression $\sup_{0<t\le\delta} \|\widehat w(t,\cdot)\|_{L^2(U)}$
for functions of the form \Eqref{eq:hatw} yields a norm which is
equivalent to $\|\cdot\|_{\delta,\alpha,1}^\sim$, thanks to
\Eqref{eq:behaviorofk}. Therefore, the energy \Eqref{eq:defenergy} is
of relevance for the discussion of functions in $\tilde X_{\delta,\alpha,1}$.

\begin{lemma}[Fundamental energy estimate for the regular initial
  value problem]
  \label{lem:energyestimates}
  Suppose that the source-term is of the form \Eqref{eq:linearinhom2}
  with the conditions \Eqref{eq:linearinhomDecay} and that the energy
  dissipation matrix \Eqref{eq:defNfinal3} is positive semidefinite on
  $(0,\delta]\times U$ for given constants $\alpha,\eta>0$.  Then, if
  $\delta>0$ is sufficiently small, there exist constants
  $C,\kappa,\gamma>0$, independent of the choice of $t_0\in
  (0,\delta]$, so that for all solutions $w$ of the regular initial
  value problem with smooth regular data at $t=t_0$, we have \be
   \label{eq:integratedEnergyEstimate}
   \aligned
    &      \|\widehat{w}(t,\cdot)\|_{L^2(U)}
    \\
    & \leq  
    C \, e^{\frac 12\kappa(t^\gamma-t_0^\gamma)}
    \Bigl(\|\widehat{w}(t_0,\cdot)\|_{L^2(U)}
    +\int_{t_0}^t
    s^{-1}\|s^{\Re\lambda_2-\alpha}
    f_0(s,\cdot)\|_{L^2(U)}ds\Bigr),
   \endaligned
  \ee 
  for all $t\in [t_0,\delta]$.
\end{lemma}

The role of the matrix $N$ in \Eqref{eq:defNfinal3} in the proof of
this result motivates the name ``energy dissipation
matrix''. Moreover, this results demonstrates the importance of the
assumption $\beta(x)>-1$ in \Eqref{eq:behaviorofk}. Namely, if
$\beta(x)\le -1$ at a point $x\in U$, then for any choice of $\alpha$
and $\eta$, the matrix $N$ would not be positive semidefinite for
small $t$ at $x$. While the energy estimate would still be true for a
given $t_0$, we would nevertheless lose uniformity of the constants
in the estimates with respect to $t_0$. We already stress at this
stage that it is this uniformity that will be crucial in the proof of
\Propref{prop:existenceSVIPlinear}.

\begin{proof}[Proof of Lemma~\ref{lem:energyestimates}]
  It will be convenient to work with the function
  \[\widetilde w:=t^{\lambda} w\]
  for some smooth periodic function $\lambda:U\rightarrow\R$, in order
  to optimize the positivity requirement on the matrix $N$ at the end
  of this proof.  Since $w$ is a smooth solution of
  \Eqref{eq:secondorderFuchsianHypForw},
  \[D^2w+2a Dw+b w-t^2 k^2 \partial_x^2 w=F[w]
  \] 
  for $t\ge t_0$ with coefficients $a$, $b$, (or $\lambda_1$ and
  $\lambda_2$), and $k$, it follows by direct computation that
  \begin{equation*}
    D^2\widetilde w+2\widetilde a D\widetilde w+\widetilde b \widetilde w-t^2
    k^2\partial_x^2 \widetilde w 
    =t^{\lambda}f_0 +F_L[\tilde w].
  \end{equation*}
  Here, $F_L[\tilde w]$ is an expression linear in $\eta \widetilde
  w$, $D\widetilde w$ and $t k\partial_x\widetilde w$ with smooth
  coefficient functions, which are\footnote{In the case that $\lambda$
    is not a constant, this is strictly speaking only true if
    $\lambda$ is not too negative.} $O(t^{\mu'})$ at $t=0$ for some
  constant $\mu'>0$. Hence, the new source-term is again of the form
  \Eqref{eq:linearinhom2} with the conditions
  \Eqref{eq:linearinhomDecay}. The coefficients of the principle part
  are given by
  \[\widetilde a=a-\lambda,\quad \widetilde b=b-2 a \lambda+\lambda^2,\]
  so that
  \[\widetilde\lambda_1=\lambda_1-\lambda,\quad
  \widetilde\lambda_2=\lambda_2-\lambda.
  \] 
  We consider the energy $E[\widetilde w]$ with respect to these 
  coefficients
  and find
  \begin{align*}
    DE[\widetilde w]=&-\kappa\gamma t^\gamma E[\widetilde w]
    +e^{-\kappa t^\gamma} 
    \int_U 2(\Re\widetilde\lambda_2-\alpha) 
    t^{2(\Re\widetilde\lambda_2-\alpha)}e[\tilde w] 
    dx\\
    &+e^{-\kappa t^\gamma} 
    \int_U  t^{2(\Re\widetilde\lambda_2-\alpha)} De[\tilde w] dx.
  \end{align*}
  Now,
  \begin{align*}
    De[\widetilde w]&=\eta^2\, \widetilde w D\widetilde w
    +D\widetilde w D^2\widetilde w
    +(1+{D{k}}/{{k}})({t k}\partial_x \widetilde w)^2
    +t^2 k^2\partial_x\widetilde w \partial_xD\widetilde w\\
    &=\eta^2\, \widetilde w D\widetilde w+D\widetilde w 
    (-2\widetilde a D\widetilde w-\widetilde b \widetilde w 
    +t^2 k^2\partial_x^2 \widetilde w+ t^{\lambda}f_0+F_L[\widetilde w])\\
    &\quad+(1+{D{k}}/{{k}})({t k}\partial_x \widetilde w)^2
    +t^2 k^2\partial_x\widetilde w \partial_xD\widetilde w\\
    &=-(\widetilde b-\eta^2) \widetilde w D\widetilde w
    -2\widetilde a (D\widetilde w)^2 
    +(1+{D{k}}/{{k}})({t k}\partial_x \widetilde w)^2
    +D\widetilde w  (t^{\lambda}f_0+F_L[\widetilde w])\\
    &\quad+t^2 k^2\partial_x^2 \widetilde w D\widetilde w
    +t^2 k^2\partial_x\widetilde w \partial_xD\widetilde w.
  \end{align*}
  When $De[\widetilde w]$ is multiplied with 
  $t^{2(\Re\widetilde\lambda_2(x)-\alpha)}$, the last
  two terms can be treated as follows
  \begin{align*}
    &t^{2(\Re\widetilde\lambda_2(x)-\alpha+1)}(
    {k}^2\partial_x^2 \widetilde w D\widetilde w
    +{k}^2\partial_x\widetilde w \partial_xD\widetilde w)\\
    &=\partial_x(t^{2(\Re\widetilde\lambda_2(x)-\alpha+1)}{k}^2
    \partial_x \widetilde w D\widetilde w)
    -2t^{2(\Re\widetilde\lambda_2(x)-\alpha)}
    (t \partial_x{k})({t k}\partial_x \widetilde w) D\widetilde w\\
    &\quad-2 t^{2(\Re\widetilde\lambda_2(x)-\alpha)}
    (\partial_x\Re\widetilde\lambda_2)
    ({t k}\ln t)({t k}\partial_x \widetilde w) D\widetilde w.
  \end{align*}
  The first term on the right vanishes after integration in space by
  virtue of periodicity on the domain $U$. Now we collect all terms of
  $DE$ as follows
  \[DE[\widetilde w]=:DE_1[\widetilde w]+DE_2[\widetilde w],\]
  where
  \begin{align*}
    DE_1[\widetilde w]:=
    e^{-\kappa t^\gamma} \int_U t^{2(\Re\widetilde\lambda_2-\alpha)}
    \Bigl(&
    (\Re\widetilde\lambda_2-\alpha)(\eta \widetilde w)^2
    +(\Re\widetilde\lambda_2-\alpha-2\widetilde a)(D\widetilde w)^2\\
    &+(\Re\widetilde\lambda_2
    -\alpha+1+\frac{D{k}}{k})({t k}\partial_x\widetilde w)^2
    -(\widetilde b/\eta-\eta)(\eta \widetilde w)D\widetilde w\\
    &-2(t\partial_x{k}
    +(\partial_x\Re\widetilde\lambda_2)
    ({t k}\ln t))({t k}\partial_x \widetilde w)D\widetilde w\Bigr)dx,
  \end{align*}
  and
  \begin{equation*}
    \begin{split}
      DE_2[\widetilde w]:=\frac 12 e^{-\kappa t^\gamma} \int_U
      t^{2(\Re\widetilde\lambda_2-\alpha)} \Bigl(& -\kappa\gamma
      t^\gamma(\eta \widetilde w)^2 -\kappa\gamma t^\gamma(D\widetilde w)^2
      -\kappa\gamma t^\gamma({t k} \partial_x \widetilde w)^2\\
      &+2(t^{\lambda}f_0+F_L[\widetilde w])D\widetilde w\Bigr)dx.
    \end{split}
  \end{equation*}
  Using the expressions of $\widetilde\lambda_{1}$ and $\widetilde\lambda_{2}$, we get
  \begin{align*}
    DE_1[\widetilde w]:=
    \int_U t^{2(\Re\lambda_2-\lambda-\alpha)}
    \Bigl(&
    -(\lambda-\Re\lambda_2+\alpha)(\eta \widetilde w)^2
    -(\Re\lambda_1-\lambda+\alpha)(D\widetilde w)^2\\
    &-(\lambda-\Re\lambda_2+\alpha-1-\frac{D{k}}{k})
    ({t k}\partial_x\widetilde w)^2\\
    &-\left((\Im\lambda_1)^2/\eta-\eta\right)
    (\eta \widetilde w)D\widetilde w\\
    &-2(t\partial_x{k}
    +\partial_x(\Re\lambda_2-\lambda)({t k}\ln t))
    ({t k}\partial_x \widetilde w)D\widetilde w\Bigr)dx.
  \end{align*}
  When we choose $\lambda=\Re\lambda_1$, as we will do now, we can
  write $DE_1$ as follows
  \begin{equation*}
    DE_1[\widetilde w]=\int_U  
    -(\widehat{\widetilde w}\cdot N\cdot \widehat{\widetilde w}^T) dx,
  \end{equation*}
  where $N$ is the energy dissipation matrix in \Eqref{eq:defNfinal3} and
  \[\widehat{\widetilde w}:=t^{\Re\widetilde\lambda_2(x)-\alpha}
  (\eta \widetilde w(t,x),D\widetilde w(t,x),t k(t,x)\partial_x\widetilde
  w(t,x)).
  \]
  Note that $\widehat{\widetilde w}=\widehat w\cdot T$ with
  $\widehat w$ from \Eqref{eq:hatw} and
  \begin{equation}
    \label{eq:trafomatrixT}
    T:=
    \begin{pmatrix}
      1 & \Re\lambda_1/\eta & (\partial_x\Re\lambda_1) (t k \ln t)/\eta\\
      0 & 1 & 0\\
      0 & 0 & 1
    \end{pmatrix}.
  \end{equation}
  The matrix $T$ is invertible and can hence be interpreted as the
  transformation matrix from the variable $w$ to the variable
  $\widetilde w$ and vice versa. Thus, if $N$ is positive semidefinite
  at all $(t,x)$, it follows that $DE_1[\widetilde w]\le 0$ for all
  $\widetilde w$, and hence for all $w$. According to the hypothesis
  of this lemma, this is the case, and we are left with
  \begin{equation*} 
    DE[\widetilde w](t)\le DE_2[\widetilde w](t).
  \end{equation*}
  Denote by $\langle\cdot,\cdot\rangle_{L^2(U)}$ the Euclidean
  $L^2$-scalar product.  Thanks to the properties of the expression
  $F_L[\widetilde w]$, we can choose $\kappa>0$ (large enough) and
  $\gamma>0$ (small enough) so that, uniformly for all $t_0$ (provided
  $\delta$ is sufficiently small),
  \[DE_2[\widetilde w](t)\le
    e^{-\kappa t^\gamma}
    \langle (0,t^{\Re\lambda_2-\alpha}f_0(t,\cdot),0),
    \widehat{\widetilde w}(t,\cdot)\rangle_{L^2(U)}.
  \] 
  
  In order to integrate the differential inequality for $E[\widetilde
  w]$ now, we use \Eqref{eq:relationEnergyNorm} and the Cauchy-Schwarz
  inequality to obtain
  \begin{equation*}
    D\left(
      \frac 12e^{-\kappa t^\gamma}\|\widehat{\widetilde w}(t,\cdot)\|_{L^2(U)}^2\right)
    \le e^{-\kappa t^\gamma}\|t^{\Re\lambda_2-\alpha}f_0(t,\cdot)\|_{L^2(U)}
    \|\widehat{\widetilde w}(t,\cdot)\|_{L^2(U)}.
  \end{equation*}
  This yields
  \[\frac d{dt}\|\widehat{\widetilde  w}(t,\cdot)\|_{L^2(U)}\le 
  \frac 12\kappa\gamma t^{\gamma-1}\|\widehat{\widetilde  w}(t,\cdot)\|_{L^2(U)}
  +t^{-1}\|t^{\Re\lambda_2-\alpha}f_0(t,\cdot)\|_{L^2(U)}.
  \]
  Then the
  Gronwall inequality implies
  \begin{equation*} 
    \|\widehat{\widetilde  w}(t,\cdot)\|_{L^2(U)}
    \le e^{\frac 12\kappa(t^\gamma-t_0^\gamma)}
    \left(\|\widehat{\widetilde  w}(t_0,\cdot)\|_{L^2(U)}
      +\int_{t_0}^t s^{-1}\|s^{\Re\lambda_2-\alpha}f_0(s,\cdot)\|_{L^2(U)}ds\right),
  \end{equation*}
  for all $t_0\le t\le\delta$. Thanks to the properties of the
  transformation matrix $T$ defined in \Eqref{eq:trafomatrixT} under
  our assumptions, one can check that there exists a
  constant $C>0$ (independent of $t_0$ and $t$) so that
  \begin{equation*} 
    \|\widehat{w}(t,\cdot)\|_{L^2(U)}\le C e^{\frac 12\kappa(t^\gamma-t_0^\gamma)}
    \left(\|\widehat{w}(t_0,\cdot)\|_{L^2(U)}
      +\int_{t_0}^t s^{-1}\|s^{\Re\lambda_2-\alpha}f_0(s,\cdot)\|_{L^2(U)}ds\right).
  \end{equation*}
\end{proof}

\begin{proof}[Proof of \Propref{prop:existenceSVIPlinear}]
  We start by assuming a smooth function $f_0$ on $t>0$, and look for
  weak solutions $w\in \tilde X_{\delta,\alpha,1}$. The first step is to
  consider a monotonically decreasing sequence
  $(\tau_n)_{n\in\N}\subset (0,\delta]$ converging to $0$. We define a
  sequence $(w_n)_{n\in\N}$ of functions on $(0,\delta]\times U$ as
  follows. For all $n\in\N$, we set $w_n(t)=0$ for all
  $t\in(0,\tau_n]$. On the interval $[\tau_n,\delta]$, we let $w_n$ be
  the unique solution of the RIVP for zero regular data at
  $t_0=\tau_n$.  The linearity of the equation and the conditions on
  the coefficients imply that $w_n$ is well-defined on the whole
  interval $(0,\delta]$, and that $w_n\in C^1([0,\delta]\times U)$. It
  is easy to see that indeed, $(w_n)\subset \tilde X_{\delta,\alpha,1}$ for
  all $\alpha>0$.  The motivation for choosing the sequence $(w_n)$ is
  that the associated functions $v_n$ (defined according to
  \Eqref{eq:splitv2}) can be hoped to behave more and more like a
  solution to the second-order Fuchsian equation obeying the
  leading-order behavior dictated by $u$, when $n$ tends to infinity;
  recall that our assumption $u\equiv 0$ represents no loss
  of generality. Hence, the sequence $(w_n)$ is expected to converge
  to a solution $w$ of the SIVP. We prove that this is the case making
  use of the energy estimates for the RIVP derived in
  \Lemref{lem:energyestimates}.

  Fix some arbitrary $m,n\in\N$ with $m\ge n$, thus
  $0<\tau_m\le\tau_n\le\delta$, and define $\xi:=w_m-w_n$. Thus, $\xi$
  is identically zero on $(0,\tau_m]$, it satisfies
  \Eqref{eq:secondorderFuchsianHypForw} with the given source-term for
  $[\tau_m,\tau_n]$, and it satisfies
  \Eqref{eq:secondorderFuchsianHypForw} with the given source-term but
  with vanishing inhomogeneity $f_0$ for $[\tau_n,\delta]$ (due to the
  linearity of the equation). Then, \Lemref{lem:energyestimates}
  implies that
  \be
    \label{eq:intermediateenergy}
    \|\widehat \xi(t,\cdot)\|_{L^2(U)}
    \begin{cases}
      =0, & \text{$t\in(0,\tau_m]$},\\
      \le C e^{\frac 12\kappa(t^\gamma-\tau_m^\gamma)}\int_{\tau_m}^t
      s^{-1}\|s^{\Re\lambda_2-\alpha}f_0(s,\cdot)\|_{L^2(U)}ds, & 
      \text{$t\in[\tau_m,\tau_n]$},\\
      \le C e^{\frac 12\kappa(t^\gamma-\tau_m^\gamma)}\int_{\tau_m}^{\tau_n}
      s^{-1}\|s^{\Re\lambda_2-\alpha}f_0(s,\cdot)\|_{L^2(U)}ds, & 
      \text{$t\in[\tau_n,\delta]$},
    \end{cases}
  \ee
  where $\widehat\xi$ is the vector-valued function associated with $\xi$
  in the same way as in \Eqref{eq:hatw}. 
  In particular, all constants here are independent of $\tau_m$ and
  $\tau_n$.
  We get
  \be
    \label{eq:estimatesource}
    s^{-\epsilon}\|s^{\Re\lambda_2-\alpha}f_0(s,\cdot)\|_{L^2(U)}
    \le 
    \|f_0\|_{\delta,\alpha+\epsilon,0},
  \ee
  for all $s\in (0,\delta]$.  Here, $\epsilon>0$ is the constant given
  in the hypothesis of this proposition.
  In total, the map $s\mapsto
  s^{-\epsilon}\|s^{\Re\lambda_2-\alpha}f_0(s,\cdot)\|_{L^2(U)}$ is a
  bounded continuous function on $(0,\delta]$ and hence is
  integrable. Hence, the function
  \begin{equation*}
    G(t):=\int_{0}^{t} s^{-1}\|s^{\Re\lambda_2-\alpha}f_0(s,\cdot)\|_{L^2(U)}ds
  \end{equation*} 
  is well-defined and finite for all $t$. Indeed, it is continuous for
  $t\in (0,\delta]$ and $\lim_{t\rightarrow 0} G(t)=0$. This implies
  that $G$ is uniformly continuous on $(0,\delta]$.  By taking the
  supremum in $t$ on the interval $(0,\delta]$ of
  \Eqref{eq:intermediateenergy} and adapting the constant $C$ if
  necessary, we show that
  \begin{equation}
    \label{eq:Cauchysequence}
    \|w_m-w_n\|_{\delta,\alpha,1}^\sim\le C|G(\tau_n)-G(\tau_m)|.
  \end{equation}
  In particular, we point out that while $C$ can depend on the choice
  of $\delta$, it is bounded for small $\delta$.  Now, since $G$ is
  uniformly continuous and $(\tau_n)$ is a Cauchy sequence with limit
  zero, it follows that $(G(\tau_n))$ is a Cauchy sequence with limit
  zero. Thus, \Eqref{eq:Cauchysequence} implies that the sequence
  $(w_n)$ is a Cauchy sequence in $(\tilde
  X_{\delta,\alpha,1},\|\cdot\|_{\delta,\alpha,1}^\sim)$, and hence
  there exists a limit function $w\in \tilde X_{\delta,\alpha,1}$.

  Now we have to check that $w$ is a weak solution of the Fuchsian
  equation, i.e.\ $\left\langle\mathcal P[w],\phi\right\rangle=0$ for
  all test functions $\phi$ according to \Eqref{eq:weakequation}. Consider an arbitrary 
  test-function $\phi$ and pick up any $n\in\N$.  The sequence
  element $w_n$ is constructed so that
  \[|\left\langle\mathcal P[w_n],\phi\right\rangle|\le 
  \int_{0}^{\tau_n} \left|\left\langle s^{\Re\lambda_2-\alpha}f_0(s,\cdot),
    \phi(s,\cdot)\right\rangle_{L^2(U)}\right|ds.
  \]
  This estimate holds since, for all $t>\tau_n$, the approximate
  solution $w_n$ satisfies  the equation
  \Eqref{eq:secondorderFuchsianHypForw}. This yields
  \begin{equation*}
    |\left\langle\mathcal P[w_n],\phi\right\rangle|  
    \le \sup_{t\in(0,\delta]}\|t \phi(t,\cdot)\|_{L^2(U)}
    \int_{0}^{\tau_n} s^{-1}\left\|s^{\Re\lambda_2-\alpha}f_0(s,\cdot)
    \right\|_{L^2(U)}ds
    =\widetilde C\, G(\tau_n),
  \end{equation*}
  for some constant $\tilde C$.  Hence,
  $\lim_{n\rightarrow\infty}\left\langle\mathcal P[w_n],\phi\right\rangle=0$.
  Since $\left\langle\mathcal P[\cdot],\phi\right\rangle$ is 
  continuous on $\tilde X_{\delta,\alpha,1}$ with respect
  to the norm $\|\cdot\|_{\delta,\alpha,1}^\sim$ for any given test
  function $\phi$ as noted earlier, it follows that
  \[\left\langle\mathcal P[w],\phi\right\rangle=0.\]
  Hence $w$ is a solution of the SIVP.

  Let us check that the limit $w$ of the sequence
  $(w_n)$ does not depend on the choice of sequence $(\tau_n)$, i.e.\
  that our solution procedure yields a unique solution (which,
  however, is not guaranteed to be the only solution of the SIVP at
  this stage of the proof). Let $(\tau_n)$ be a monotonically
  decreasing sequence in $(0,\delta]$ with limit $0$, $(w_n)$ the
  corresponding sequence of approximate solutions and $w$ its limit;
  the same for another monotonically decreasing sequence
  $(\widetilde\tau_n)$ in $(0,\delta]$ with limit $0$, sequence of
  approximate solutions $(\widetilde w_n)$ and limit $\widetilde
  w$. Now, we take the union of the two sequences $(\tau_n)$ and
  $(\widetilde\tau_n)$ and sort the new sequence $(\widehat\tau_n)$,
  so that it becomes a monotonically decreasing Cauchy sequence with
  limit $0$. Then \Eqref{eq:Cauchysequence} shows that
  $\|w_n-\widetilde w_m\|_{\delta,\alpha,1}^\sim\rightarrow 0$ for
  $n,m\rightarrow\infty$. Hence
  \[\|w-\widetilde w\|_{\delta,\alpha,1}^\sim\le
  \|w-w_n\|_{\delta,\alpha,1}^\sim
  +\|w_n-\widetilde w_m\|_{\delta,\alpha,1}^\sim
  +\|w_m-\widetilde w_m\|_{\delta,\alpha,1}^\sim\rightarrow 0,\]
  and so $w=\widetilde w$.

  So far, we have shown that for all 
  $f_0\in X_{\delta,\alpha+\epsilon,0}\cap C^\infty((0,\delta]\times U)$,
  there exists a weak solution $w\in \tilde X_{\delta,\alpha,1}$, and
  that $w$ is independent of the choice of sequence $(\tau_n)$.  Then
  the solution operator
  \begin{gather*}
    \mathbb H:
    \tilde X_{\delta,\alpha+\epsilon,0}\cap C^\infty((0,\delta]\times U)
    \rightarrow \tilde X_{\delta,\alpha,1},\quad
    f_0\mapsto w,
  \end{gather*}
  is well-defined. It is clearly linear, and we derive the estimate
  \Eqref{eq:continuityHPDE} now. From \Eqref{eq:Cauchysequence}, we
  get
  \[\|w\|_{\delta,\alpha,1}^\sim\le \|w_1\|_{\delta,\alpha,1}^\sim+
  C G(\delta).
  \] 
  We can estimate $\|w_1\|_{\delta,\alpha,1}$ as follows. Because $w_1$
  is a solution of the RIVP with zero regular data at $t_0=\tau_1$,
  estimate \Eqref{eq:integratedEnergyEstimate} yields
  \begin{equation*}
    \|\widehat w_1(t,\cdot)\|_{L^2(U)}\le
    C e^{\frac 12\kappa(t^\gamma-\tau_m^\gamma)} \int_{\tau_1}^t
    s^{-1}\|s^{\Re\lambda_2-\alpha}f_0(s,\cdot)\|_{L^2(U)}ds
  \end{equation*}
  for all $t\in[\tau_1,\delta]$. For all $t\in(0,\tau_1]$, we have
  $\|\widehat w_1(t,\cdot)\|_{L^2(U)}= 0$. This shows that
  $\|w_1\|_{\delta,\alpha,1}\le C G(\delta)$ with some adapted $C$,
  and hence, absorbing the factor $2$ into the constant, we get
  \be
    \label{eq:estimatew}
    \|w\|_{\delta,\alpha,1}^\sim\le C G(\delta).
  \ee
   Now, the estimate
  \Eqref{eq:estimatesource} gives
  \[
  G(\delta)\le \frac 1\epsilon \delta^\epsilon
  \|f_0\|_{\delta,\alpha+\epsilon,0}^\sim.
  \]  
  Using this together with \Eqref{eq:estimatew} yields
  \Eqref{eq:continuityHPDE} on the subset
  $\tilde X_{\delta,\alpha+\epsilon,0}\cap C^\infty((0,\delta]\times U)$ of
  $\tilde X_{\delta,\alpha+\epsilon,0}$. Now we proceed in the same way as in
  the proof of \Propref{prop:Hestimate} in order to
  extend this operator and the validity of \Eqref{eq:continuityHPDE}
  to the full space. For any given test function $\phi$, the
  expression $\left\langle\mathcal P[w],\phi\right\rangle$ is
  continuous with respect to $w$ and $f_0$ (in their respective
  norms). This is sufficient to show that the extended operator $\mathbb H$
  maps to a weak solution $w$.

  Finally, estimate \Eqref{eq:continuityHPDE} allows us to show
  uniqueness. Assume that we have two solutions $w$ and $\widetilde w$
  for the same source-term. Then $w-\widetilde w$ is a solution of the
  same equation with $f_0=0$. Hence, \Eqref{eq:continuityHPDE} implies
  that
  \[\|w-\widetilde w\|_{\delta,\alpha,1}^\sim\le 0,
  \] 
  and so uniqueness is established. 
\end{proof}


\subsection{Non-linear theory in the spaces
  \texorpdfstring{$\tilde X_{\delta, \alpha,2}$}
  {tilde X{delta,alpha,2}} and 
  \texorpdfstring{$X_{\delta, \alpha,\infty}$}
  {X{delta,alpha,infty}}}
\label{sec:nonlineartheory}

\paragraph{The general non-linear theory}
The well-posedness theory of the previous section, where we restrict
to the space $\tilde X_{\delta, \alpha,1}$, has certain
limitations. First, the statement that the solution of the Fuchsian
equation $w$ is an element of $\tilde X_{\delta, \alpha,1}$ yields
particularly weak information about the behavior of the first spatial
derivative at $t=0$. It would be advantageous if we were able to prove
the stronger statement $w\in X_{\delta, \alpha,1}$, possibly under
stronger assumptions. Second, it turns out that we need to require a
Lipschitz property of the source-term for the general non-linear case
which rules out natural non-linearities, for instance quadratic ones,
if we only control the first derivatives of the solution.  In the case
of one spatial dimension, as we always assume in the whole paper, it
is sufficient to increase regularity to the space $\tilde X_{\delta,
  \alpha,2}$. It is then clear how to proceed to $\tilde X_{\delta,
  \alpha,k}$ with arbitrary $k\in\N$.  Nevertheless, in some
applications \cite{BeyerLeFloch2}, the space $\tilde X_{\delta,
  \alpha,k}$ imposes too strong a restriction, due to the weak control
of the highest spatial derivative. This problem can be avoided by
formulating the theory in the space $X_{\delta,\alpha,\infty}$.

\begin{lemma}[Existence of solutions of the linear singular
  initial value problem in $\tilde X_{\delta,\alpha,2}$]
  \label{lem:existenceSVIPlinear2}
  Let us make the same assumptions as listed in the beginning
  of \Sectionref{sec:Linear1Statement} with $f_a\equiv 0$ for
  $a=1,2,3$ (for simplicity).  Then there exists a unique solution
  $w\in \tilde X_{\delta,\alpha,2}$ of the singular initial value
  problem for given $\delta,\alpha>0$ provided:
  \begin{enumerate}
  \item The energy dissipation matrix \Eqref{eq:defNfinal3}
    is positive definite at each
    $(t,x)\in (0,\delta)\times U$ for a constant $\eta>0$.
  \item The source-term function $f_0$ is in
    $X_{\delta,\alpha+\epsilon,1}$ for some $\epsilon>0$.
  \end{enumerate}
  Then, the solution operator
  \[\mathbb H: X_{\delta,\alpha+\epsilon,1}
  \rightarrow \tilde X_{\delta,\alpha,2},\quad f_0\mapsto w,
  \] 
  is
  continuous and there exists a finite constant $C_\epsilon>0$ so that
  \begin{equation*} 
    \|\mathbb H[f_0]\|_{\delta,\alpha,2}^\sim \le \delta^{\epsilon}
    C_\epsilon \|f_0\|_{\delta,\alpha+\epsilon,1},
  \end{equation*}
  for all $f_0$.  The constant $C_\epsilon$ is bounded for all small
  $\delta$.
\end{lemma}

Analogous results hold for systems and for general linear source terms
of the form \Eqref{eq:linearinhom2} with non-vanishing functions
$f_1$, $f_2$ and $f_3$ obeying decay conditions analogous to
\Eqref{eq:linearinhomDecay} also for the first derivatives. Moreover,
the result can be generalized to an arbitrary number $k$ of
derivatives, i.e.\ to solutions in the space $\tilde
X_{\delta,\alpha,k}$. For $k\ge 3$ in one spatial dimension, the
Sobolev inequalities imply that the weak derivatives can be identified
with classical derivatives. Hence the solution $w\in \tilde
X_{\delta,\alpha,3}$ of the weak form of the equation
\Eqref{eq:weakequation} is then a classical solution of
\Eqref{eq:secondorderFuchsianHyp} with $v=w$, i.e.\ $u\equiv 0$.

\begin{proof}
  One sees immediately that (the generalization to systems of)
  \Propref{prop:existenceSVIPlinear} applies directly to the system of
  equations for the unknowns $(w_0,w_1,w_2)$ with $w_0:=w$, $w_1=Dw$
  and $w_2:=t^{\epsilon_1}\partial_x w$, where $\epsilon_1>0$ can be
  any sufficiently small constant. We cannot choose $\epsilon_1=0$
  since this would lead to a source-term which is not consistent with
  the hypotheses of \Propref{prop:existenceSVIPlinear} in the
  following.  For this system the energy dissipation matrix has the
  required properties, if it has the required properties for the
  original equation for $w$ and if we assume the same constant
  $\alpha$ for the equations for $w_0$, $w_1$ and $w_2$. However, the
  energy dissipation matrix must positive definite instead of positive
  semidefinite due to the presence of the non-vanishing constant
  $\epsilon_1$. One obtains existence and uniqueness in a space
  $\tilde X_{\delta,\alpha,1}$ for vector-valued functions
  $(w_0,w_1,w_2)$. The thus obtained space $\tilde
  X_{\delta,\alpha,1}$ for vector-valued functions $(w_0,w_1,w_2)$
  equals the space $\tilde X_{\delta,\tilde\alpha,2}$ for the original
  scalar function $w$ where $\tilde\alpha$ differs from $\alpha$ by
  the arbitrarily small constant $\epsilon_1$.
\end{proof}
It is important to note that, when we repeat the proof for an arbitrary
number of derivatives $k$, the quantity $\tilde\alpha$ can be chosen
arbitrarily close to $\alpha$ irrespective of the choice of $k$.

\begin{proposition}[Existence of solutions of the non-linear singular
  initial value problem in $\tilde X_{\delta,\alpha,2}$]
  \label{prop:existenceSVIPnonlinear2}
  Suppose that we can choose $\alpha>0$ so that the energy dissipation
  matrix \Eqref{eq:defNfinal3} is positive definite at each
  $(t,x)\in (0,\delta)\times U$ for a constant $\eta>0$.  Suppose that
  $u\equiv 0$ and that the operator $F$ has the following Lipschitz
  continuity property: For a constant $\epsilon>0$ and all
  sufficiently small $\delta$, the operator $F$ maps $\tilde
  X_{\delta,\alpha,2}$ into $X_{\delta,\alpha+\epsilon,1}$ and,
  moreover, for each $r>0$ there exists $\widehat C>0$ (independent of
  $\delta$) so that \be
  \label{eq:lipschitzFPDE2}
  \|F[w]-F[\widetilde w]\|_{\delta,\alpha+\epsilon,1} \le \widehat C \, 
  \|w-\widetilde w\|_{\delta,\alpha,2}^\sim
  \ee
  for all $w,\widetilde w\in \overline{B_r(0)}\subset
  \tilde X_{\delta,\alpha,2}$.
  Then,    
  there exists a unique
  solution $w\in \tilde X_{\delta,\alpha,2}$ of the singular initial value
  problem.
\end{proposition}

The generalization of this result to arbitrarily many  
derivatives is again straightforward.

\begin{proof}[Proof of \Propref{prop:existenceSVIPnonlinear2}:]  
  Similar to
  \Sectionref{sec:firstorderFuchsian}, we define the operator $\mathbb
  G:=\mathbb H\circ F$, and argue in the same way as in
  \Propref{prop:iterationsequence} and
  \Theoremref{th:fixedpointtheorem} that under the hypothesis, this
  operator is a contraction on closed and bounded subsets of $\tilde
  X_{\delta,\alpha,2}$ if $\delta$ is a sufficiently small. Hence the
  iteration sequence defined by $w_{j+1}=\mathbb G[w_j]$ for $j\ge 1$
  and, say, $w_1=0$ converges to a fixed point $w\in \tilde
  X_{\delta,\alpha,2}$ with respect to the norm
  $\|\cdot\|_{\delta,\alpha,2}^\sim$. Because of the properties of
  $\mathbb H$, a fixed point of $\mathbb G$ is a solution of the
  SIVP. Hence, we have shown existence of solutions. Uniqueness can be
  shown as follows. Given any other solution $\tilde w$ in $\tilde
  X_{\delta,\alpha,2}$, it is a fixed point of the iteration
  $w_{j+1}=\mathbb G[w_j]$. Because $\mathbb G$ is a contraction,
  there, however, only exists one fixed point, and hence $\tilde w=w$.
\end{proof}

\begin{proposition}[Existence of solutions of the non-linear singular
  initial value problem in $X_{\delta,\alpha,\infty}$]
  \label{prop:existenceSVIPnonlinearinfty}
  Suppose that we can choose $\alpha>0$ so that the energy dissipation
  matrix \Eqref{eq:defNfinal3} is positive definite at each
  $(t,x)\in (0,\delta)\times U$ for a constant $\eta>0$.  Suppose that
  $u\equiv 0$ and that the operator $F$ has the following Lipschitz
  continuity property: For a constant $\epsilon>0$, every sufficiently
  small $\delta>0$ and every non-negative integer $k$, the operator
  $F$ maps $X_{\delta,\alpha,k+1}$ into $X_{\delta,\alpha+\epsilon,k}$
  and, moreover, for each $r>0$, there exists $\widehat C>0$
  (independent of $\delta$) so that
  \begin{equation}
    \label{eq:LipschitzCinfty}
    \|F[w]-F[\widetilde w]\|_{\delta,\alpha+\epsilon,k} \le \widehat C \, 
    \|w-\widetilde w\|_{\delta,\alpha,k+1}^\sim
  \end{equation}
  for all $w,\widetilde w\in \overline{B_r(0)}\cap
  X_{\delta,\alpha,k+1}\subset \tilde X_{\delta,\alpha,k+1}$.  Then,
  there exists a unique solution $w\in X_{\delta,\alpha,\infty}$ of
  the singular initial value problem. 
\end{proposition}

Here, $\overline{B_r(0)}$ is defined with respect to the norm
$\|\cdot\|^\sim_{\delta,\alpha,k+1}$.  We note that the constant
$\widehat C$ is allowed to depend on $k$. Note that the Lipschitz
estimate involves the norm $\|\cdot\|_{\delta,\alpha,k+1}^\sim$, while
the elements for which this estimates needs to be satisfied are
required to be only in the subspace $X_{\delta,\alpha,k+1}$ of $\tilde
X_{\delta,\alpha,k+1}$. The main advantage of this result over the
finite differentiability case is that we only need to check that $F$
maps $X_{\delta,\alpha,k+1}$ into $X_{\delta,\alpha+\epsilon,k}$ for
all $k$, instead of the stronger statement that $F$ maps $\tilde
X_{\delta,\alpha,k+1}$ into $X_{\delta,\alpha+\epsilon,k}$ (which
would of course, however, only need to hold for finitely many $k$).

\begin{proof}
  We first generalize \Lemref{lem:existenceSVIPlinear2} to the case
  $f_0\in X_{\delta,\alpha,\infty}$. Then the solution of the linear
  equation $\mathbb H[f_0]\in X_{\delta,\alpha,\infty}$. Since $F$
  maps $X_{\delta,\alpha,\infty}$ to itself, the operator $\mathbb
  G:=\mathbb H\circ F$ maps $X_{\delta,\alpha,\infty}$ to
  itself. Hence, the same iteration as in the proof of
  \Propref{prop:existenceSVIPnonlinear2} leads to an iteration
  sequence $(w_j)\subset X_{\delta,\alpha,\infty}$. Let $k\ge 1$ be
  arbitrary. According to the hypothesis, $F$ maps
  $X_{\delta,\alpha,k}$ into
  $X_{\delta,\alpha+\epsilon,k-1}$. \Lemref{lem:existenceSVIPlinear2}
  implies that $\mathbb H$ maps $X_{\delta,\alpha+\epsilon,k-1}$ to
  $\tilde X_{\delta,\alpha,k}$. Hence $\mathbb G$ can be consider as a
  map $X_{\delta,\alpha,k}\rightarrow \tilde X_{\delta,\alpha,k}$. We
  can consider $X_{\delta,\alpha,k}$ to be a subset of $\tilde
  X_{\delta,\alpha,k}$. Then, if we choose $\delta$ small enough, the
  restriction of $\mathbb G$ to $X_{\delta,\alpha,k}\cap
  \overline{B_r(0)}\subset \tilde X_{\delta,\alpha,k}$ is a
  contraction with respect to the norm
  $\|\cdot\|^\sim_{\delta,\alpha,k}$ due to
  \Eqref{eq:LipschitzCinfty}. This implies that for an appropriate
  choice of $r$, we have that $(w_j)$ is a Cauchy sequence in
  $X_{\delta,\alpha,\infty}\cap \overline{B_r(0)}\subset \tilde
  X_{\delta,\alpha,k}$ with respect to the norm
  $\|\cdot\|^\sim_{\delta,\alpha,k}$. Hence, it converges to a limit
  $w_{(k)}\in \tilde X_{\delta(k),\alpha,k}$. We have written
  $\delta(k)$ now instead of $\delta$ in order to stress that $\delta$
  does depend on $k$. We get such a limit function $w_{(k)}\in \tilde
  X_{\delta(k),\alpha,k}$ for all $k\ge 1$. It is straightforward to
  check that $w_{(k_1)}(t)=w_{(k_2)}(t)$ for all two integers
  $k_1,k_2\ge 1$ for all $t\in
  (0,\min\{\delta(k_1),\delta(k_2)\}]$. Hence, it follows that
  $w_{(k)}\in X_{\delta(k+1),\alpha,k}$ (without tilde!) for all $k$,
  and this means that $w_{(k)}\in X_{\delta(k),\alpha,k}$ after
  possibly decreasing $\delta(k)>0$ sufficiently for all $k$. So for
  any given $k$, the limit $w_{(k)}$ is in the range of the operator
  $\mathbb G$, and thus $w_{(k)}$ is the unique fixed point of
  $\mathbb G$ and so the unique solution of the equation in
  $X_{\delta(k),\alpha,k}$. However, it is not obvious at this point
  whether we are forced to choose $\delta(k)\rightarrow 0$ for
  $k\rightarrow\infty$, and we are left with demonstrating that this
  is not the case. As soon as we have this, we have constructed the
  solution $w\in X_{\delta,\alpha,\infty}$ for some $\delta>0$. This,
  however, requires only standard arguments for symmetric hyperbolic
  equations. 
  Hence, we find that we can choose $\delta=\delta(3)$ (for one spatial
  dimension).
\end{proof}

\paragraph{The (standard) singular initial value problem}
The following discussion is devoted to particular choices of the
function $u$ motivated by the heuristics introduced
in \Sectionref{sec:firstorderFuchsian}.  Consider the case that $u$ is given by
\begin{equation}
  \label{eq:SIVPDefV}
  u(t,x) = u_0(t,x):=\begin{cases}
    u_*(x)\,t^{-a(x)}\ln t+u_{**}(x)\,t^{-a(x)}
    & \quad a^2=b,\\
    \hat u_*(x)\,t^{-\lambda_1(x)}+\hat u_{**}(x)\,t^{-\lambda_2(x)},
    & \quad a^2 \neq b,
  \end{cases}
\end{equation}
with $\hat u_*, \hat u_{**}$ given by \Eqref{eq:transitiondata} and
with asymptotic data $u_*,u_{**}\in H^3(U)$.  In this case, we will
speak of the \textbf{standard singular initial value
  problem}\footnote{In order to simplify the language, we often speak
  of the \textit{singular initial value problem} if there is no risk of
  confusion.}.  Note that this means that for all $t>0$, the map
$u(t,\cdot)$ and all its time derivatives are in $H^3(U)$.

\begin{theorem}[Well-posedness of the standard singular initial value problem
  in $\tilde X_{\delta,\alpha,2}$]
  \label{th:well-posednessSIVP}
  Given arbitrary asymptotic data $u_*,u_{**}\in H^3(U)$, the
  standard singular initial value problem admits a unique solution
  $w\in \tilde X_{\delta,\alpha,2}$ for $\alpha,\delta>0$, provided
  $\delta$ is sufficiently small and the following conditions hold:
  \begin{enumerate}
  \item{Positivity condition.} Suppose that we can choose $\alpha>0$
    so that the energy dissipation matrix \Eqref{eq:defNfinal3} is
    positive definite at each $(t,x)\in (0,\delta)\times U$ for a
    constant $\eta>0$.
  \item{Lipschitz continuity property.} For the given $\alpha>0$, the
    operator $F$ satisfies the Lipschitz continuity property stated in
    \Propref{prop:existenceSVIPnonlinear2} for all asymptotic data
    $u_*,u_{**}\in H^3(U)$ for some $\epsilon>0$.
  \item{Integrability condition.} The constants $\alpha$ and
    $\epsilon$ satisfy
    \begin{equation}
      \label{eq:intbeta}
      \alpha+\epsilon<2(\beta(x)+1)-\Re(\lambda_1(x)-\lambda_2(x)), 
      \qquad  x \in U.
    \end{equation}  
  \end{enumerate}
\end{theorem}

An analogous theorem can be formulated for the $C^\infty$-case based
on \Propref{prop:existenceSVIPnonlinearinfty}. In this case, the
asymptotic data $u_*$, $u_{**}$ must be in $C^\infty(U)$ and the
Lipschitz condition must be substituted by the condition of
\Propref{prop:existenceSVIPnonlinearinfty}. The unique solution $w$ of
the singular initial value problem is then an element of
$X_{\delta,\alpha,\infty}$.

We note that there might be room for improvements in the finite
differentiability case $k=2$, since three derivatives
of the asymptotic data yield control of only two derivatives of the
solution.

\begin{proof}
  We can apply \Propref{prop:existenceSVIPnonlinear2} if we are able
  to control the additional contribution of the term $L[u]$ which has
  to be considered as part of the source-term. It has no contribution
  to the Lipschitz estimate \Eqref{eq:lipschitzFPDE2}, but we have to
  guarantee that under these hypotheses, $L[u]\in
  X_{\delta,\alpha+\epsilon,1}$ for the given constant
  $\epsilon$. This is indeed the case if \Eqref{eq:intbeta} holds.
\end{proof}

\begin{example}
  Consider the second-order hyperbolic Fuchsian equation
  \[D^2 v-\lambda Dv-t^2\partial_x^2 v=0,
  \] 
  with a constant $\lambda$. This is the Euler-Poisson-Darboux
  equation. In the standard notation it is
  \[
  \partial_t^2 v-\partial_x^2 v=\frac 1t(\lambda-1)\partial_t v.
  \]
  Note that $\lambda=1$ is the standard wave equation, and in
  this case, the standard singular initial value problem reduces to
  the standard Cauchy problem.
  \begin{enumerate}
  \item Case $\lambda\ge 0$. With our notation, we have
    $\lambda_1=0$, $\lambda_2=-\lambda$, $\beta\equiv 0$,
    $\nu\equiv 1$ and $f\equiv 0$. The positivity condition of the
    energy dissipation matrix \Eqref{eq:defNfinal3} is satisfied
    precisely for $\alpha\ge 1-\lambda$ and all sufficiently
    small $\eta>0$. The integrability condition \Eqref{eq:intbeta} is
    satisfied precisely for $\lambda< 2-\alpha$. Hence, our
    previous proposition implies that the singular initial value
    problem is well-posed, provided 
    \be
      \label{eq:wellposednessEPD}
      0\le\lambda<2.
    \ee
    Namely, in this case there
    exists a solution $w$ in $\tilde X_{\delta,\alpha,2}$ for some $\alpha>0$ for
    arbitrary asymptotic data in $H^3(U)$.    
  \item Case $\lambda<0$. With our notation, we have
    $\lambda_1=|\lambda|$, $\lambda_2=0$, $\beta\equiv 0$,
    $\nu\equiv 1$ and $f\equiv 0$. The positivity condition of the
    energy dissipation matrix \Eqref{eq:defNfinal3} is satisfied
    precisely for $\alpha\ge 1-|\lambda|$ and all sufficiently
    small $\eta>0$. The integrability condition \Eqref{eq:intbeta} is
    satisfied precisely for $|\lambda|< 2-\alpha$. Hence, our
    previous proposition implies that the singular initial value
    problem is well-posed, provided
    \begin{equation*} 
      -2<\lambda<0.
    \end{equation*}
    Namely, in this case there exists a solution $w$ in
    $\tilde X_{\delta,\alpha,2}$ for some $\alpha>0$ for arbitrary asymptotic
    data in $H^3(U)$.
  \end{enumerate}
  Now, it turns out that general smooth solutions to the
  Euler-Poisson-Darboux equation can be expressed explicitly by a
  Fourier ansatz in $x$ and by Bessel functions in $t$. It is then
  easy to check that \Eqref{eq:wellposednessEPD} (and similarly for
  $\lambda<0$) is sharp: While for $0\le\lambda<2$, all
  solutions of the equation behave consistently with the two-term
  expansion at $t=0$, this is not the case for $\lambda\ge 2$
  for general asymptotic data. Hence the singular initial value
  problem is not well-posed for $\lambda\ge 2$. This is
  completely consistent with our heuristic discussion
  in \Sectionref{sec:asymptoticsolutions}. Namely, if
  $\lambda=2$, the assumption that the source-term
  $t^2\partial_x^2 v$ is negligible at $t=0$ fails since it is of
  the same order in $t$ at $t=0$ as the second leading-order term.
  However, we can see in the proof of
  \Theoremref{th:well-posednessSIVP} that in the special case $u_*=0$
  (and arbitrary $u_{**}$), the integrability condition
  \Eqref{eq:intbeta} can be relaxed.  For this special choice of data,
  solutions to the singular initial value problem exist even for
  $\lambda\ge 2$.
\end{example}

\paragraph{Singular singular initial value problems with asymptotic
  solutions of order $j$}
One of the main aims of this paper is to study the well-posedness of
the standard singular initial value problem just discussed. In this
sense, we can be satisfied with
\Theoremref{th:well-posednessSIVP}. However, it turns out that, often in
applications, the three conditions in this theorem cannot be
satisfied simultaneously. While it is often possible to find constants
$\alpha$ and $\epsilon$ in accordance with the second and third
condition, it can turn out that the corresponding choice of $\alpha$
is too small to make the energy dissipation matrix positive
definite. The following trick can sometimes solve this problem.

For the following discussion, we need to bring together results
from Sections~\ref{sec:firstorderFuchsian} and
\ref{sec:hyperbolicfuchsianequations}, and we are forced to
distinguish between operators now which for the sake of simplicity
have carried the same name so far.  Consider some asymptotic data $u_*$ and
$u_{**}$ and define the function $u_0$ as in
\Eqref{eq:SIVPDefV}. We write $\hat f$ for the source-term in
\Eqref{eq:secondorderFuchsian} and continue to write $f$ for the
source-term in \Eqref{eq:secondorderFuchsianHyp}. Accordingly, we
write $\hat F[w]:=\hat f[u_0+w]$ and $F[w]:=f[u_0+w]$, so that e.g.\
\[\hat F[w]=F[w]+t^2 k^2\partial_x^2 (u_0+w).\]
In the same way, we distinguish between operators $L$ and $\hat L$.
Now we make the same assumptions on $\hat F$ as in
\Propref{prop:iterationsequence} in \Sectionref{sec:asymptsol}. If the
asymptotic data is in $H^{m_1}(U)$ for some positive integer $m_1$,
then the function $u_j$, referred to as $w_j$ in
\Propref{prop:iterationsequence}, with $u_1=0$ is well-defined in
$X_{\delta,\tilde\alpha,l,m_1-2(j-1)}$ for some $\tilde\alpha>0$ and
all $j$ with $m_1/2+1\ge j\ge 1$, provided $F$ maps
$X_{\delta,\tilde\alpha,l,m}$ to
$X_{\delta,\tilde\alpha+\epsilon,l-1,m-1}$ for all $m\le m_1$, for
some integer $l\ge 1$ and for some $\epsilon>0$.

Now let us choose the leading-order function $u$ as
\begin{equation}
  \label{eq:leadingtermHigherOrder}
  u(t,x)=u_0(t,x)+u_j(t,x),
\end{equation} 
for $j$ in the range given above. We refer to the singular initial
value problem based on this choice of leading-order term as
\textbf{singular initial value problem with asymptotic solutions of
  order $j$}. For $j=1$, it reduces to the standard singular initial
value problem; hence we will focus on the case $j\ge 2$ in the
following. Note that, if $w$ is a solution of the singular initial
value problem of order $j$, it is also a solution of the standard
singular initial value problem. However, if there is only one solution
$w$ of the singular initial value problem with asymptotic solutions of
order $j$ for given asymptotic data, it does not mean that $w$ is the
only solution of the standard initial value problem for the same
asymptotic data.

It can be seen easily that the remainder $w$ of a solution of the
singular initial value problem with asymptotic solutions of order $j$
satisfies the equation
\[L[w]=F_j[w]:=F[u_j+w]-F[u_{j-1}]+t^2 k^2\partial_x^2(u_j-u_{j-1})\]
for $j\ge 2$. Now thanks to \Theoremref{th:orderofsequence} we have 
\[
u_{j}-u_{j-1}\in X_{\delta,\tilde\alpha+(j-2)\kappa\epsilon_0,l,m_1-2(j-1)},
\] 
for all $\kappa<1$. Hence it is reasonable to restrict to remainders
$w\in X_{\delta,\tilde\alpha+(j-2)\kappa\epsilon_0,l,m_1-2(j-1)-1}$ in
the following. One finds easily that this means that $F_j[w]\in
X_{\delta,\tilde\alpha+(j-1)\kappa\epsilon_0,l-1,m_1-2(j-1)-2}$ if
$2(\beta(x)+1)>\kappa\epsilon$ for all $x\in U$ by using similar
arguments as in \Theoremref{th:orderofsequence}. This gives us hope
that we can apply \Propref{prop:existenceSVIPnonlinear2} with
\begin{equation}
  \label{eq:increasealpha}
  \alpha:=\tilde\alpha+(j-2)\kappa\epsilon.
\end{equation}
The effect of the ansatz \Eqref{eq:leadingtermHigherOrder} is a value
of $\alpha$ which increases $\tilde\alpha$ by
$(j-2)\kappa\epsilon$. Namely if $m_1$ is sufficiently large, we can
choose $j$ large enough so that the energy dissipation matrix,
evaluated for $\alpha$, can become positive definite. The main
prize that we pay with this approach is that the asymptotic data must
be sufficiently regular and that we must live with a loss of
regularity which is stronger the larger $j$ is.

For the statement of the following theorem, we need the following notation. 
For all $w\in X_{\delta,\alpha,k}$ (or $w\in \tilde
X_{\delta,\alpha,k}$ respectively), we introduce the functions
$E_{\delta,\alpha,k}[w]:(0,\delta]\rightarrow\R$ (or $\tilde
E_{\delta,\alpha,k}[w]:(0,\delta]\rightarrow\R$ respectively) which
are defined in the same way as the respective norms, but the supremum
in $t$ has not been evaluated yet. In particular, this means that
$E_{\delta,\alpha,k}[w]$ (or $\tilde E_{\delta,\alpha,k}[w]$) is a
bounded continuous function on $(0,\delta]$.

\begin{theorem}[Well-posedness of the singular initial value problem
  with asymptotic solutions of higher-order in $\tilde
  X_{\delta,\alpha,2}$]
  \label{th:WellPosednessHigherOrderSIVP}
  Given any integer $j\ge 2$ and any asymptotic data $u_*,u_{**}\in
  H^{m_1}(U)$ with $m_1=2j+1$, there exists a unique solution $w\in
  \tilde X_{\delta,\alpha,2}$ of the singular initial value problem
  with asymptotic solutions of order $j$ for some $\alpha>0$, provided
  \begin{enumerate}
  \item $F$ maps $\tilde X_{\delta,\tilde\alpha,m_1}$ into
    $X_{\delta,\tilde\alpha+\epsilon,m_1-1}$ for all asymptotic data
    $u_*,u_{**}\in H^{m_1}(U)$ for some $\epsilon>0$ and
    $\tilde\alpha$ given by \Eqref{eq:increasealpha} for an arbitrary
    $\kappa<1$.
  \item The characteristic speed satisfies
    \[2(\beta(x)+1)>\kappa\epsilon \quad\text{for all }x\in U\]
    for the same constant $\kappa$ chosen earlier.
  \item $F$ satisfies the following Lipschitz condition: for each
    $r>0$ there exists a constant $C>0$ (independent of $\delta$) so
    that
    \[E_{\delta,\tilde\alpha+\epsilon,1}[F[w]-F[\widetilde w]](t)
    \le C \tilde E_{\delta,\tilde\alpha,2}[w-\widetilde w](t)\]
    for all $t\in (0,\delta]$ and
    for all $w,\widetilde w\in \overline{B_r(0)}\subset
    \tilde X_{\delta,\tilde\alpha,2}$.
  \item The energy dissipation matrix \Eqref{eq:defNfinal3} (evaluated
    with $\alpha$) is positive definite at each $(t,x)\in
    (0,\delta)\times U$ for a constant $\eta>0$.
  \end{enumerate}
\end{theorem}

The third condition above is meaningful since both sides of
the inequality are continuous and bounded functions on $(0,\delta]$. Note that this
theorem can be formulated without difficulty for the $C^\infty$-case and leads to a simpler statement.

\begin{proof}
  Under the first hypothesis, it can be shown similar to
  \Propref{prop:iterationsequence} that the $j$-th element of the
  iteration sequence $u_j$ is in
  $X_{\delta,\tilde\alpha,m_1-2(j-1)}$. As in
  \Theoremref{th:orderofsequence} it can be demonstrated that
  $u_j-u_{j-1}$ is in
  $X_{\delta,\tilde\alpha+(j-2)\kappa\epsilon,m_1-2(j-1)}$ for all
  $j\ge 2$. Now let $w\in
  \tilde X_{\delta,\tilde\alpha+(j-2)\kappa\epsilon,m_1-2(j-1)-1}$. Under the
  second hypothesis, it follows that $F_j[w]\in
  X_{\delta,\tilde\alpha+j \kappa\epsilon,m_1-2(j-1)-2}$. Now, in
  order to apply \Propref{prop:existenceSVIPnonlinear2} with $F$
  substituted by $F_j$, it is necessary to choose $m_1=2j+1$. The
  operator $F_j$ satisfies the Lipschitz condition of
  \Propref{prop:existenceSVIPnonlinear2} if the third hypothesis is
  satisfied. Thanks to the fourth assumption, we can now apply 
  \Propref{prop:existenceSVIPnonlinear2}.
\end{proof}


\section*{Acknowledgements}

The authors were partially supported by the Agence Nationale de la Recherche
(ANR) through the grant 06-2-134423 entitled {\sl ``Mathematical Methods in General Relativity''} (MATH-GR).
The first draft of this paper was written during the year 2008--2009
when the first author (F.B.) was an ANR postdoctoral fellow at the Laboratoire J.-L. Lions.  
The second author (P.L.F.) is grateful to the Erwin Schr\"odinger Institute,
Vienna, where this paper was completed during the program {\sl ``Quantitative Studies of Nonlinear Wave Phenomena'',} 
organized by P.C. Aichelburg, P. Bizo\'n, and W. Schlag.



\begin{thebibliography}{1}

\bibitem{ABL} 
\auth{Amorim P., Bernardi C., and LeFloch P.G.,} 
Computing Gowdy spacetimes via spectral evolution in future and past directions,  
\jou{Class. Quantum Grav.} 26 (2009), 1--18.

\bibitem{AnderssonRendall}
\auth{Andersson L. and Rendall A.D.,}  
Quiescent cosmological singularities,
\jou{Commun. Math. Phys.} 218 (2001), 479--511.

\bibitem{BLSS} 
\auth{Barnes A.P., LeFloch P.G., Schmidt B.G., and Stewart J.M.,}
The Glimm scheme for perfect fluids on plane-symmetric Gowdy spacetimes, 
\jou{Class. Quantum Grav.} 21 (2004), 5043--5074.

\bibitem{BergerMoncrief} 
\auth{Berger B.K., Moncrief V.,}
Numerical investigations of cosmological singularities, 
\jou{Phys. Rev. D} 48 (1993), 4676.

\bibitem{BergerChruscielMoncrief} 
\auth{Berger B.K., Chru\'sciel P., and Moncrief V.,}
On asymptotically flat spacetimes with $G_2$-invariant Cauchy surfaces, 
\jou{Ann. Phys.} 237 (1995), 322--354. 

\bibitem{BergerChruscielIsenbergMoncrief} 
\auth{Berger B.K., Chru\'sciel P., Isenberg J., and Moncrief V.,}
Global foliations of vacuum spacetimes with $T^2$ isometry, 
\jou{Ann. Phys.} 260 (1997), 117--148. 

\bibitem{BeyerLeFloch2} 
\auth{Beyer F. and LeFloch P.G.,}
Second-order hyperbolic Fuchsian systems. Gowdy spacetimes and the Fuchsian numerical algorithm, 
ArXiv:1006.2525.

\bibitem{BeyerLeFloch3} 
\auth{Beyer F. and LeFloch P.G.,}
in preparation. 

\bibitem{Chrusciel} 
\auth{Chru\'sciel P.,} 
On spacetimes with $U(1) \times U(1)$ symmetric compact Cauchy surfaces,
\jou{Ann. Phys.} 202 (1990), 100--150.
 
\bibitem{ChruscielIsenbergMoncrief} 
\auth{Chru\'sciel P., Isenberg J., and Moncrief V.,} 
Strong cosmic censorship in polarized Gowdy spacetimes, 
\jou{Class. Quantum Grav.} 7 (1990), 1671--1680. 

\bibitem{EardleyMoncrief} 
\auth{Eardley D. and Moncrief V.,}
The global existence problem and cosmic censorship in general relativity, 
\jou{Gen. Relat. Grav.} 13 (1981), 887--892.  

\bibitem{Gowdy73}
\auth{Gowdy R.H.,} 
Vacuum space-times with two parameter spacelike isometry groups and
compact invariant hypersurfaces: Topologies and boundary conditions,
\jou{Ann. Phys.} 83 (1974), 203--241.

\bibitem{IsenbergMoncrief} 
\auth{Isenberg J. and Moncrief V.,}
Asymptotic behavior of the gravitational field and the nature of singularities in Gowdy spacetimes, 
\jou{Ann. Phys.} 99 (1990), 84--122.

\bibitem{KichenassamyRendall}
Kichenassamy S. and Rendall A.D., 
Analytic description of singularities in Gowdy spacetimes,
\jou{Class. Quantum Grav.} 15 (1998), 1339--1355.

\bibitem{LeFlochRendall} 
\auth{LeFloch P.G. and Rendall A.D.,} 
A global foliation of Einstein-Euler spacetimes with Gowdy-symmetry on $T^3$, 
Preprint Series ESI 2242, Erwin Schr\"odinger Institute, Vienna. See also arXiv:1004.0427v1.

\bibitem{LeFlochStewart} 
\auth{LeFloch P.G. and Stewart J.M.,}
Shock waves and gravitational waves in matter spacetimes with Gowdy symmetry, 
\jou{Port. Math.} 62 (2005), 349--370.

\bibitem{LeFlochStewart2} 
\auth{LeFloch P.G. and Stewart J.M.,}
The characteristic initial value problem for plane symmetric spacetimes with weak regularity, 
Preprint Series ESI 2243, Erwin Schr\"odinger Institute, Vienna. See also arXiv:1004.2343v1. 

\bibitem{Moncrief} 
\auth{Moncrief V.,}
Global properties of Gowdy spacetimes with $T^3 \times \RR$ topology, 
\jou{Ann. Phys.} 132 (1981), 87--107.

\bibitem{Rendall00}
\auth{Rendall A.D.,}
Fuchsian analysis of singularities in {G}owdy spacetimes beyond analyticity,
\jou{Class. Quantum Grav.} 17 (2000), 3305--3316.

\bibitem{RendallWeaver} 
\auth{Rendall A.D. and Weaver M.,}
Manufacture of Gowdy spacetimes with spikes,
\jou{Class. Quantum Grav.} 18 (2001), 2959--2975. 

\bibitem{Ringstrom2} \auth{Ringstr\"om H.,} 
Asymptotic expansions close to the singularity in Gowdy spacetimes,
\jou{Class. Quant. Grav.} 21 (2004); S305--S322. 

\bibitem{Ringstrom4} 
\auth{Ringstr\"om H.,}  
Curvature blow up on a dense subset of the singularity in T3-Gowdy,
\jou{J. Hyperbolic Diff. Eqs.} 2 (2005), 547--564. 

\bibitem{Ringstrom6} 
\auth{Ringstr\"om H.,}  
Strong cosmic censorship in $T^3$-Gowdy spacetimes, 
\jou{Ann. Math.} 170 (2009), 1181--1240.  

\end{thebibliography}
\end{document}